\documentclass[11pt]{article}

\newif\ifauthornote
\authornotefalse

\usepackage[margin=1in]{geometry}
\usepackage{amsmath,amssymb,amsthm,mathtools}
\usepackage{enumitem}
\usepackage{microtype}
\usepackage{xspace}
\usepackage{xcolor}
\definecolor{darkgreen}{RGB}{0,100,0}
\usepackage[bookmarks=true,pdfstartview=FitH,colorlinks,linkcolor=darkgreen,filecolor=darkgreen,citecolor=darkgreen,urlcolor=darkgreen]{hyperref}

\usepackage[capitalize,noabbrev]{cleveref}

\newtheorem{theorem}{Theorem}
\newtheorem{lemma}[theorem]{Lemma}

\newtheorem{fact}[theorem]{Fact}
\theoremstyle{definition}

\theoremstyle{remark}
\newtheorem{remark}[theorem]{Remark}

\newcommand{\name}{{\sf Otter}\xspace}
\newcommand{\R}{\mathbb{R}}

\newcommand{\Yside}{\mathcal{Y}}
\newcommand{\Xside}{\mathcal{X}}
\newcommand{\Util}{\mathsf{Util}}

\newcommand{\netYSale}{-\Delta_Y}
\newcommand{\netXSale}{-\Delta_X}

\newcommand{\playerids}{\mathcal{I}}
\newcommand{\reportv}[1]{\ensuremath{\widetilde{v}_{#1}}\xspace}
\newcommand{\reportq}[1]{\ensuremath{\widetilde{q}_{#1}}\xspace}
\newcommand{\reportu}[1]{\ensuremath{\widetilde{u}_{#1}}\xspace}
\newcommand{\reportr}[1]{\ensuremath{\widetilde{r}_{#1}}\xspace}

\ifauthornote
\newcommand{\elaine}[1]{{\color{red}[elaine: #1]}\xspace}
\else
\newcommand{\elaine}[1]{}
\fi

\title{Tight Characterizations of AMM Mechanism Design \\under Weak Efficiency\thanks{The order of authorship is randomized. The random order was generated using the
American Economic Association's Author Randomization Tool, and can be verified with
confirmation code \texttt{kmt-XxqFQO\_w} at
\url{https://www.aeaweb.org/journals/policies/random-author-order/search?RandomAuthorsSearch\%5Bsearch\%5D=kmt-XxqFQO_w}.}}
\newcommand{\authorbreak}{%
  \end{tabular}\par\vspace{0.5em}%
  \begin{tabular}[t]{c}%
}
\author{
Chengkai Zhang\thanks{Rutgers University. \texttt{zhangchengkai2002@gmail.com}.}
\and Lirong Xia\thanks{Rutgers University. \texttt{lirong.xia@rutgers.edu}.}
\and Elaine Shi\thanks{Carnegie Mellon University. \texttt{elaineshi@cmu.edu}.}
\and Mengqian Zhang\thanks{Carnegie Mellon University. \texttt{mengqianzhang@cmu.edu}.}
\authorbreak Qishen Han\thanks{Rutgers University. \texttt{Hnickc2017@gmail.com}.}
\and Yuhao Li\thanks{Columbia University and Carnegie Mellon University. \texttt{yuhaoli@cs.columbia.edu}.}
}
\date{}

\begin{document}
\maketitle

\begin{abstract}
Miner extractable value (MEV) poses a persistent challenge for blockchain ecosystems.
It arises from the fact that a block builder can exploit its unilateral
control over transaction inclusion and ordering to extract profit,
potentially at users' expense. Recent work has explored 
\emph{consensus-aided mechanism design} to mitigate MEV.
This approach combines generic consensus-layer guarantees
with application-layer mechanisms that use these guarantees to address
application-specific MEV problems. 

Prior work~\cite{ammmechdesign} exemplifies this approach for automated market makers (AMMs).
When the consensus layer offers fair sequencing, which orders transactions according to their
arrival times, their proposed mechanism achieves incentive compatibility
(IC) and weak local efficiency (wLE)~\cite{ammtrilemma}. Informally, IC rules out profitable strategic deviations by users and builders,
including builders with intrinsic trading demand. Meanwhile, wLE ensures that after execution, 
any order willing to trade at the initial price cannot benefit from further trade with the pool at the ending marginal price.

In this work, we resolve the main open question left by prior works~\cite{ammmechdesign,ammtrilemma}: whether IC and wLE can
be achieved under plain consensus (e.g., current Ethereum), without fair sequencing. 
We prove that no deterministic, anonymous two-asset batch AMM can achieve both properties. 
The impossibility holds even when the consensus layer guarantees censorship resilience and block space is
unlimited, establishing the first separation between censorship resilience
and fair sequencing in the context of decentralized mechanism design.

Our impossibility proof combines three strategic capabilities: creating Sybil
identities, over-reporting budgets, and misreporting values. 
We further show the tightness of our impossibility result: removing any one from the strategy space restores feasibility under censorship resilience. 
Specifically, we construct a mechanism achieving IC and wLE for each restricted setting. 

\end{abstract}

\section{Introduction}
\label{sec:intro}

Blockchains have made it possible to build financial applications whose
rules are enforced by a smart contract rather than a trusted intermediary.
A prominent application is the Automated Market Maker (AMM): a contract 
that holds a liquidity pool of two assets, $X$ and $Y$, and allows users
to exchange one for the other according to a prescribed pricing curve.
For example, the widely adopted constant-product rule keeps the product of
the pool's reserves fixed. If the pool initially holds $x_0$ units of
$X$ and $y_0$ units of $Y$, a user who pays $y>0$ units of $Y$ receives
$x$ units of $X$ such that
\[
    (x_0-x)(y_0+y)=x_0y_0.
\]
This rule automates liquidity provision and pricing, allowing
users to trade against the pool without matching buyers and
sellers. 
It does not, however, make the
resulting market immune to strategic behaviors, some of which are particularly salient in decentralized systems.

A central concern is the phenomenon of Miner Extractable Value
(MEV), whereby a block builder exploits its
unilateral control over the inclusion and ordering of transactions
within a block to extract profit. 
Numerous prior studies~\cite{highfreq-dex, just-in-time-defi, torres2021frontrunner, darkforest, li2023demystifying} 
have documented pervasive MEV extraction in practice. 
In a sandwich attack, for example,
a builder buys an asset immediately before a user's purchase and sells it
immediately afterward, profiting from the price increase caused by the
user's trade. MEV not only harms ordinary users but also encourages
builders to offer favorable transaction ordering through private agreements
with traders, which has reinforced centralization in the builder
market~\cite{buildercentral}.

To mitigate MEV, recent works~\cite{ammmechdesign,ammtrilemma,otter} 
have explored the noval approach of \emph{consensus-aided mechanism design},
which combines consensus-layer guarantees with application-layer
mechanism design. Rather than
asking the consensus layer to understand and eliminate every
application-specific MEV opportunity, one can ask it to provide
generic guarantees that an application uses in its mechanism design. 
Chan, Wu, and Shi~\cite{ammmechdesign} exemplify this approach for AMMs.
Under a \emph{fair-sequencing} consensus guarantee, their mechanism achieves incentive compatibility (IC) 
as understood in decentralized mechanism design, a substantially stronger requirement than standard game-theoretic IC: 
no user, builder, or builder with an intrinsic demand to trade acting as a user (which we henceforth refer to as a \emph{builder-as-user}) 
can improve its utility by misreporting its trading demand or valuation, 
or by submitting additional orders through pseudonymous identities (i.e., a Sybil attack).
In particular, this IC guarantee effectively prevents sandwich attacks: an attacker with no
intrinsic trading demand cannot earn risk-free profit by inserting front-running
and back-running orders. 

Their mechanism~\cite{ammmechdesign} also achieves a natural notion of efficiency called \emph{weak local efficiency} (wLE)~\cite{ammtrilemma}:
an \emph{eligible order}, one willing to trade at the initial price, 
may remain partially or entirely unfilled only
if further trade with the pool at the ending marginal price offers no 
strictly positive surplus. Indeed, Li, Shi, and Zhang~\cite{ammtrilemma}
show that stronger efficiency notions, including local efficiency for
all orders and Pareto optimality, are incompatible with IC even under fair sequencing.

The fair-sequencing model abstracts a consensus layer that orders
transactions according to their arrival times at the sequencer.
It is motivated by the community's recent efforts to develop decentralized sequencers~\cite{decentral-seq,espresso-seq}
and practical fair-ordering protocols~\cite{orderfair00,orderfair01,orderfair02,orderfair03,orderfair04}.
From the builder's perspective, fair sequencing prevents it from
censoring transactions or freely choosing their ordering.\footnote{Under fair
sequencing, a builder cannot alter the arrival times of other users'
orders. It may delay its own orders but cannot assign them earlier arrival
times. A strategic participant with a faster network may nevertheless 
submit an order after observing another user's transaction and have its order arrive first. 
Thus, fair sequencing alone does not eliminate front-running.}
Without these additional restrictions, a builder can choose both the transaction inclusion 
and ordering of its block. We call the latter \emph{plain consensus model}, which more closely reflects the single-builder block
construction used in mainstream consensus protocols.

An important open question left by prior works~\cite{ammmechdesign,ammtrilemma} is whether IC and wLE can be achieved in the plain consensus, 
without relying on fair sequencing. 
This question probes the necessity of consensus for the consensus-aided mechanism
design approach: \emph{is consensus-layer support essential, or can application-layer
mechanism design alone provide the same guarantees?}

\subsection{Our Results and Contributions}
Following prior work, we study a two-asset \emph{batch AMM} that processes
the orders in a block jointly. We require \emph{pricing curve conformity}:
the pool's state after the batch must lie on its prescribed pricing curve,
and all tokens paid out by the pool must be distributed among the orders
in the batch. This requirement applies to the batch as a whole; individual
trades need not follow the curve in isolation. The mechanism by \cite{ammmechdesign} establishes on this model, 
and their mechanism satisfies the pricing curve conformity condition.\footnote{Instead, in the familiar model of AMMs such as Uniswap, swaps are executed one at a time: each swap moves the pool's reserves from one point on the pricing curve to another. The reserves must therefore lie on the curve after every individual swap. See more discussion on this relaxition in \Cref{sec:relatedwork}.} 

Our main contribution is to resolve the main open question left in~\cite{ammmechdesign,ammtrilemma}: we establish a strong impossibility of IC + wLE AMM mechanisms under plain consensus. Indeed, the
impossibility persists even if the consensus layer
guarantees censorship resilience and block space is unlimited.
\begin{theorem}
No deterministic AMM satisfying individual rationality, pricing
curve conformity, and anonymity can simultaneously satisfy incentive
compatibility and weak local efficiency in the plain consensus. The impossibility continues to hold even if the consensus layer
guarantees censorship resilience and block space is unlimited.
\end{theorem}

Censorship resilience, compared to plain consensus, removes the builder's powerful ability to exclude transactions. 
Concerns about centralization in block building have motivated efforts to actively 
pursue censorship resilience at the consensus layer, including Ethereum's planned adoption 
of FOCIL~\cite{focil} and exploration into other approaches~\cite{mcp,latencyCR}.
Our result nevertheless shows that, within our model,  merely protecting inclusion alone is insufficient 
to reconcile IC with even weak local efficiency.

Recently, Shi, Zhang, Chung, and Li~\cite{otter} studies AMM mechanism design in a new paradigm 
that allows \emph{surplus redistribution}: some tokens output
by the pool may be directed to the broader community rather
than to orders in the current batch. In our terminology, this
relaxes the exact-settlement requirement in pricing curve
conformity, while still requiring the pool's final reserves
to lie on the prescribed curve. They propose a mechanism called~\name 
that achieves IC and maximizes social welfare under censorship-resilient consensus and uncongested block space. They also show that such a mechanim is impossible when block space is finite and the builder can censor orders.

Conceptually, our result establishes the necessity of relaxing
the exact-settlement requirement: even under censorship-resilient
consensus, retaining this requirement makes IC incompatible with even wLE.
This constitutes a fundamental insight for AMM mechanism
design, extending beyond the particular construction
of~\name: any mechanism seeking MEV-resilience and reasonable efficiency achievement in the
same setting must depart from exact settlement.

Technically, our result establishes the first separation between
censorship resilience and fair sequencing: within our model, IC and wLE are incompatible under the former,
whereas prior work achieves both under the latter~\cite{ammmechdesign}.
To the best of our knowledge, this is the first such separation
in decentralized mechanism design, showing that fair sequencing is strictly more powerful than censorship resilience in a precise sense.\footnote{In the lieterature, a number of impossibility results have been established
in decentralized mechanism design, including for transaction fee
mechanisms~\cite{foundation-tfm,collusion-tfm} and AMM mechanisms~\cite{ammtrilemma,otter,ammclinch}. These impossibility results either hold only under plain consensus: censorship
resilience together with uncongested block space restores
feasibility~\cite{foundation-tfm,collusion-tfm,otter}, or hold under fair
sequencing or even in permissioned settings~\cite{ammtrilemma,ammclinch}.
In contrast, our impossibility persists under censorship resilience
and uncongested block space but disappears under fair sequencing,
thereby separating these two consensus guarantees for the first time.}

\paragraph{Tightness with respect to strategic capabilities.}
Looking ahead, our impossibility proof combines three devitation capabilities: creating Sybil
identities, over-reporting budgets, and misreporting values.
We show that all of them are necessary for the impossibility, i.e., we have to use all of them simultaneously. 
Imposing any of the following restrictions on strategy space restores feasibility of IC and wLE, together with other natural properties, under censorship resilience.
\begin{enumerate}[leftmargin=6mm,itemsep=3pt]
    \item \emph{No Sybil identities.}
    Each user submits at most one order, but may still misreport its
    value and budget, including over-reporting its budget, and choose
    either trading direction. An equal-share \emph{hand-in-hand}
    mechanism achieves all desired properties in this setting
    (Theorem~\ref{thm:tsf-no-sybil}).

    \item \emph{No aggregate budget over-reporting.}
    Users may create Sybil identities and misreport their values, but the
    aggregate declared budget in their true trading direction cannot
    exceed their true budget. A \emph{proportional hand-in-hand}
    mechanism achieves all desired properties
    (Theorem~\ref{thm:tsf-no-overreport}).

    \item \emph{Public values.}
    A user's identities in its true trading direction must share
    its public value. Budgets remain private and may be
    overreported. A \emph{layered mechanism}, processing orders in
    increasing order of their asks, achieves all desired properties
    (Theorem~\ref{thm:tsf-public-values}).
\end{enumerate}
The latter two settings still permit submitting fake bids in both directions,
with unrestricted asks and declared budgets for opposite-direction
orders. In particular, the public-value result does not assume that
budgets are public or that opposite-direction asks are fixed by the
user's public value. 

We note that when aggregate budget over-reporting across the Sybil identities 
in the true trading direction is prohibited (setting for the item 2 above), a recent work \cite{ammclinch} obtains an IC mechanism with even stronger efficiency guarantee. 
Their AMM mechanism, inspired by the classic work of adaptive clinching~\cite{adaptiveclinch}, satisfies both IC and weak Pareto optimality for eligible bids, 
which they also prove to be essentially the best possible efficiency guarantee in this setting. Our mechanism, by contrast, is substantially simpler and is intended primarily to establish feasibility rather than to offer an economically compelling design.

Technically, all three of our feasibility constructions share a two-sided matching component. 
We compare the aggregate reported budgets of eligible orders in common units 
at the initial exchange rate. Eligible orders on the smaller-budget side 
trade their full budgets at that price, fixing their outcomes. 
The remaining orders are handled by a one-sided mechanism 
using an augmented compensation curve that incorporates these matched trades 
before following the pool's pricing curve. 
The idea of initial-price matching already appears in~\cite{ammmechdesign}, and
the augmented-curve formulation is explicit in~\cite{otter}.
Our three constructions make explicit a useful division of the design
problem into two components: initial-price matching across the two sides,
and a one-sided allocation rule tailored to the setting.
Lemma~\ref{lem:tsf-settlement} establishes the corresponding settlement
guarantees for the properties studied here, although incentives still require
analyzing the complete two-sided mechanism.  More broadly, we expect initial-price matching to become a standard
starting point for batch-AMM mechanism design, which allows concentration on the one-sided problem.

\subsection{Additional Related Work}\label{sec:relatedwork}
\paragraph{Sequential execution vs. batch settlement.}
Ferreira and Parkes~\cite{credible-ex} study verifiable sequencing rules for AMMs that execute trades sequentially, 
requiring the pool to remain on its prescribed pricing curve after every trade.
They show that, even for constant-product AMMs, every sequencing rule admits an initial pool state 
and a transaction profile under which the builder can obtain risk-free profit. 
Thus, in their model, it is impossible to have an AMM satisfying arbitrage resilience, a property weaker than IC. 
Later work by Chan, Wu, and Shi~\cite{ammmechdesign} relaxes the constraint by considering batch AMMs, 
where the pricing curve need only be preserved after the entire batch. 
This flexibility allows their mechanism to achieve arbitrage resilience under plain consensus, circumventing the earlier impossibility.
From this perspective, our work asks whether batch settlement also permits the stronger guarantee of IC. 
Since a mechanism that rejects every order is trivially incentive compatible, we additionally require weak local efficiency.
Our impossibility shows that, within our model, even batch settlement cannot reconcile IC with wLE.

Several routes to feasibility nevertheless remain. Our three constructions recover IC and wLE by restricting the strategy space. 
\name~\cite{otter} demonstrates the potential of relaxing exact settlement through surplus redistribution 
under censorship resilience. 
Alternatively, fair sequencing provides stronger consensus-layer support under which IC and wLE can coexist~\cite{ammmechdesign}. 

\paragraph{Batch clearing at a uniform price.}
A line of work has explored batch trading, particularly batch clearing 
at a uniform price~\cite{batchamm00,batchamm01,cow,ramseyer2023augmenting}. 
By eliminating price differences within a batch, uniform pricing prevents internal arbitrage 
and standard sandwich attacks within the batch. 
However, uniform pricing alone does not guarantee incentive compatibility: 
when the clearing price depends on orders within the batch, participants may benefit from misreporting 
or manipulating which orders are included. 
For example, Zhang et al.~\cite{batchnotic} show that builders can extract MEV 
by selectively excluding orders from a batch and characterize their profit-maximizing strategies. 
Li, Shi, and Zhang~\cite{ammtrilemma} establish a further limitation: 
IC, wLE, and uniform pricing are incompatible even under weak fair sequencing.

\paragraph{Other MEV mitigations.}
Researchers have also explored other MEV mitigations. 
Encrypted-mempool designs~\cite{encryptmempool00,encryptmempool01,batchdec00,trx,batchdec01,batchdec02,batchdec03,batchdec04} 
keep transaction contents private until
finalization, limiting front-running based on pending transactions.
Private-submission services such as Flashbots Protect~\cite{flashbotsProtect} 
keep transactions out of the public mempool to reduce exposure to front-running and sandwich attacks.
Order-flow auctions such as MEV-Share~\cite{flashbotsMEVShare} selectively
disclose transaction information to searchers and return part of the
resulting MEV to users. Proposer-builder separation, as implemented through
MEV-Boost~\cite{flashbotsMEVBoost}, allows proposers to obtain blocks and
payments from competing builders. These approaches address different
sources or consequences of MEV under their respective information and
trust assumptions. Our work studies whether an AMM mechanism can
jointly guarantee IC and wLE under specified consensus guarantees.

\section{Model}
\elaine{state we consider deterministic amm mech earlier? maybe some LB can be for randomized}

Imagine a pool that holds two crypto-assets (also called tokens) named $X$ and $Y$, respectively.
Users can trade with the pool through two types
of orders, called sell-$Y$ and sell-$X$ orders, respectively.

\paragraph{\texorpdfstring{Sell-$Y$ orders.}{Sell-Y orders.}}
A real user $i$ of sell-$Y$ profile has type $(v_i,q_i)\in\R_{\ge0}^2$, indicating
that the user 
has a {\it budget} of $q_i \geq 0$ units of $Y$ to spend on purchasing $X$, 
and its reservation value of each unit of $Y$ is $v_i \geq 0$, measured in units of $X$.   
A sell-$Y$ identity reports a per-unit ask and a declared budget
$(\reportv{i},\reportq{i})\in\R_{\ge0}^2$.  Truthful reporting means
$(\reportv{i},\reportq{i})=(v_i,q_i)$.

\paragraph{\texorpdfstring{Sell-$X$ orders.}{Sell-X orders.}}
Symmetrically, a real user $j$ of sell-$X$ profile has type $(u_j,r_j)\in\R_{\ge0}^2$, where $r_j$ is
its budget in units of $X$, and $u_j$ is its
reservation value for one unit of $X$, measured in units of $Y$.  
A sell-$X$ identity reports
$(\reportu{j},\reportr{j})\in\R_{\ge0}^2$, and truthful reporting means
$(\reportu{j},\reportr{j})=(u_j,r_j)$.

\subsection{A Concave AMM Pricing Curve}
\label{sec:curve}

Suppose the pool initially holds finite reserves $x_0>0$ and $y_0>0$ of
tokens $X$ and $Y$, respectively.  Let $\widehat y$ denote the curve's
effective $Y$-payout capacity, where $0<\widehat y\le y_0$.
Trades are governed by a pricing curve
\[
    F:(-\widehat y,+\infty)\longrightarrow(-\infty,x_0).
\]
A signed amount $y\in(-\widehat y,+\infty)$ denotes net
$Y$ sent into the AMM, and $F(y)$ denotes net $X$ sent out of the AMM.  Thus
$y>0$ describes selling $Y$ to the AMM and receiving $X$, while $y<0$
describes withdrawing $Y$ while depositing $X$.
We assume:
\begin{enumerate}[label=(A\arabic*),leftmargin=3em]
    \item $F(0)=0$;
    \item $F$ is continuous, strictly increasing, and concave;
    \item $F$ is differentiable at $0$ with $0<F'(0)<\infty$;
    \item $F$ has the lower-endpoint behavior $\lim_{p\downarrow-\widehat y}F(p)=-\infty$.
\end{enumerate}

Define the initial prices
\[
    \sigma_0:=F'(0)
    \quad\text{($X$ per $Y$)},
    \qquad
    \rho_0:=\frac{1}{F'(0)}=\frac{1}{\sigma_0}
    \quad\text{($Y$ per $X$)}.
\]

For the sell-$Y$ direction, the mechanism uses the restriction of $F$ to
$[0,+\infty)$. 
The pool pays
out $F(y)\in[0,x_0)$ units of $X$ upon receiving $y\ge0$ units of $Y$,
and the initial marginal compensation is $\sigma_0$.
For the sell-$X$ direction, 
let
\[
    H:[0,+\infty)\longrightarrow[0,\widehat y),
    \qquad
    H(x):=-F^{-1}(-x).
\]
The pool pays out $H(x)$ units of $Y$ upon receiving $x\ge0$
units of $X$.
Its initial marginal compensation is $\rho_0 = 1/\sigma_0$.
Assumptions (A2) and (A4) ensure that $F^{-1}(-x)$ exists for every $x\ge0$.
Since $F$ is increasing and concave, $H$ is increasing and concave.
Thus the mechanism uses $F|_{[0,+\infty)}$ and $H$ as its two one-sided
compensation curves.  Both are increasing and concave and have unbounded
nonnegative input domains.

\paragraph{Constant product as a special case.}
For a constant-product AMM with initial reserves $(x_0,y_0)$ and
$k=x_0y_0$, take $\widehat y=y_0$.  Then
\[
    F(y)
    =x_0-\frac{k}{y_0+y},
    \qquad y>-y_0.
\]
It follows that
\[
    F'(y)=\frac{k}{(y_0+y)^2}>0,
    \qquad
    F''(y)=-\frac{2k}{(y_0+y)^3}<0,
\]
and
\[
    \sigma_0=\frac{x_0}{y_0},
    \qquad
    \rho_0=\frac{y_0}{x_0}.
\]

In other words, the pool pays out 
$x_0-\frac{k}{y_0+y}$ units of $X$ for receiving
$y \ge 0$ units of $Y$, 
and it pays out 
  $-F^{-1}(-x) = y_0-\frac{k}{x_0+x}$ units of $Y$ for receiving $x \ge 0$ units of $X$.
Hence constant product is a special case of our generalized concave pricing
curve formulation.

\subsection{AMM Mechanism}
An AMM mechanism receives as input a \emph{list} of orders, each of either
sell-$Y$ or sell-$X$ type, and decides an outcome $(x_i, y_i)$ for each order:
\begin{itemize}[leftmargin=6mm,itemsep=1pt]
\item 
For a sell-$Y$ order $(\reportv{i}, \reportq{i})$, it means that the user
obtains $x_i$ units in $X$ in exchange for $y_i$
units of $Y$. 
Reported-budget constraints require that $0 \leq y_i \leq \reportq{i}$.
\item 
For a sell-$X$ order $(\reportu{j},\reportr{j})$, it means that the user
obtains $y_i$ units in $Y$ in exchange for $x_i$ 
units of $X$.  
Reported-budget constraints require that $0 \leq x_i \leq \reportr{i}$.
\end{itemize}

Unless otherwise noted, we require that an AMM mechanism satisfy the following natural properties: 
\begin{itemize}[leftmargin=6mm,itemsep=1pt]
\item 
{\it Individual rationality.} 
Under truthful reporting, every user obtains nonnegative utility.

\item 
{\it Pricing curve conformity.}
Let $\Yside$ and $\Xside$ be the sets of sell-$Y$ and sell-$X$ orders, respectively. 
Define the aggregate outcome by
$$
\begin{aligned}
    S_Y&:=\sum_{i\in\Yside}y_i,&
    R_X&:=\sum_{i\in\Yside}x_i,\\
    S_X&:=\sum_{j\in\Xside}x_j,&
    R_Y&:=\sum_{j\in\Xside}y_j.
\end{aligned}
$$
Here, $S_Y$ and $R_X$ are the total $Y$ sale and $X$ receipt of the sell-$Y$ orders, 
while $S_X$ and $R_Y$ are the total $X$ sale and $Y$ receipt of the sell-$X$ orders.
Denote the net amount of $Y$ sent to the pool by $z:=S_Y-R_Y$. The corresponding net amount of $X$ paid out by the pool is
$R_X-S_X$. 
Pricing curve conformity requires
\begin{equation}
    F(z)=R_X-S_X.
    \label{eq:nc-aggregate}
\end{equation}

\item 
{\it Anonymity.}
We say an AMM mechanism is anonymous if its on-chain rules decides outcomes, 
including the allocation and payment, based solely on the valuations and budgets reported in the bids,
as well as their relative ordering within the block. Any identity-related information (e.g., public
keys) is ignored. This anonymity property is both desirable and satisfied by real-world AMMs. 
In particular, it is consistent with the core decentralization principles underlying modern blockchains.

\end{itemize}

\subsection{Strategy Space and Utility}
\label{sec:strategy}

\paragraph{Strategy space.}
A truthful user always reports its true type. 
A strategic user, however, can not only misrepresent its valuation and budget,
but also create an arbitrary finite number of identities,
and inject sybil bids. 
The declared valuation and budget associated with each sybil identity can be
arbitrary finite nonnegative reals. Moreover, the bids submitted by a
strategic user need not be in the same direction as its true type.
For example, a sell-$Y$ user may submit sell-$X$ bids, and vice versa. 

An honest builder orders bids according to a prescribed rule, for example, in the order in which it receives them. 
The builder remains free to deviate from this rule, 
and following it is an incentive requirement, not a restriction enforced by consensus. 
In general, a strategic builder can arbitrarily order the bids in the block, 
and inject sybil bids. 
If the strategic builder also has intrinsic demand, then it
can also misreport its valuation and budget. 
Note that the mechanism observes only the declared valuations and budgets,
and it is not aware which identities belong to the same real user or builder-as-user.

\paragraph{Utility.}
For the finite set $\playerids$ of identities controlled by one real user, define
aggregate token flows
\begin{equation}
\begin{aligned}
    S_Y(\playerids)&:=\sum_{i\in\playerids\cap\Yside} y_i,
    &R_X(\playerids)&:=\sum_{i\in\playerids\cap\Yside} x_i,\\
    S_X(\playerids)&:=\sum_{i\in\playerids\cap\Xside} x_i,
    &R_Y(\playerids)&:=\sum_{i\in\playerids\cap\Xside} y_i.
\end{aligned}
\label{eq:portfolio-flows}
\end{equation}
Here, $S_T$ denotes the amount of token $T \in \{X, Y\}$ supplied by the user and $R_T$
denotes the amount of token $T$ received by the user.  Its net token changes
are
\begin{equation*}
    \Delta_X(\playerids):=R_X(\playerids)-S_X(\playerids),
    \qquad
    \Delta_Y(\playerids):=R_Y(\playerids)-S_Y(\playerids).
\end{equation*}

A true sell-$Y$
user of type $(v,q)$ has utility
\begin{equation}
    \Util^Y(\playerids;v,q)
    :=
    \begin{cases}
       -\infty,
          &\text{if} \ \netYSale(\playerids)>q\ \text{or}\
           \bigl(\netYSale(\playerids)<0\ \text{and}\
                 \Delta_X(\playerids)<0\bigr) \\
       \Delta_X(\playerids)+v\Delta_Y(\playerids),
          &\text{o.w.}
    \end{cases}
    \label{eq:portfolio-utility-Y}
\end{equation}
Thus a true sell-$Y$ user receives catastrophic utility if its realized net
sale in $Y$ exceeds the true budget $q$, or if its aggregate trade is strictly
in the opposite direction of its desire, namely, if it obtains a strict net
gain in $Y$ while incurring a strict net loss in $X$.
Otherwise its utility is computed in the normal quasilinear way from
its aggregate net trade.  In particular, the utility function does not
penalize a sell-$Y$ user merely for obtaining a net gain in $Y$ when it does
not lose $X$; this includes a potential free-lunch outcome with weak net gains
in both tokens.  Nor does it penalize the user if it gross-sells more than $q$
units of $Y$ through some identities provided opposite-direction identities
buy back enough $Y$ that its final net sale does not exceed $q$.

Symmetrically, for a true sell-$X$ user of type $(u,r)$,
its utility is defined as:
\begin{equation*}
    \Util^X(\playerids;u,r)
    :=
    \begin{cases}
       -\infty,
          &\text{if} \ \netXSale(\playerids)>r\ \text{or}\
           \bigl(\netXSale(\playerids)<0\ \text{and}\
                 \Delta_Y(\playerids)<0\bigr) \\
       \Delta_Y(\playerids)+u\Delta_X(\playerids),
          &\text{o.w.}
    \end{cases}
\end{equation*}
Thus a true sell-$X$ user receives utility $-\infty$ if it net-sells more
than $r$ units of $X$, or if it obtains a strict net gain in $X$ while
incurring a strict net loss in $Y$.

\subsection{Desired Properties}

We would like the AMM mechanism to satisfy
the following desired properties.

\begin{enumerate}[label=\textnormal{(\arabic*)},leftmargin=8mm]

    \item \textbf{Incentive compatibility (a.k.a. truthfulness).}
Incentive compatibility (IC) requires that acting honestly is a dominant strategy. 
Specifically, a user cannot benefit from misreporting either its valuation or its budget, 
nor from creating additional identities and injecting Sybil bids. 
A builder is additionally incentivized to order bids in the prescribed manner. 
In particular, the builder or builder-as-user cannot increase its utility through any combination of reordering, misreporting, and injecting Sybil bids.

\item \textbf{Weak local efficiency.}
Let $\sigma_0$ be the initial spot price in units of $X$ per unit of $Y$
    and let $\rho_0=1/\sigma_0$ be the initial spot price in units of $Y$ per
    unit of $X$. 
A bid is considered {\it eligible} if its ask is no greater than the initial
spot price. 
In other words, a sell-$Y$ order $(\reportv{i}, \reportq{i})$ is eligible iff 
$\reportv{i} \leq \sigma_0$;
and a sell-$X$ order $(\reportu{j}, \reportr{j})$ 
is eligible if $\reportu{j} \leq \rho_0$. 

Let $\sigma_e$ and $\rho_e=1/\sigma_e$ be the corresponding
    final marginal prices after execution. Weak local efficiency (wLE) requires the following (economic) efficiency condition 
for eligible bids:

    For any eligible sell-$Y$ order $(\reportv{i}, \reportq{i})$, at least one of the following conditions holds:
    \begin{equation}
        y_i=\reportq{i},
        \qquad
        y_i<\reportq{i}\ \text{ and }\ \reportv{i}\ge \sigma_e.
        \label{eq:WLE-Y}
\end{equation}
    For any eligible sell-$X$ order $(\reportu{j},\reportr{j})$, at least one of the following conditions holds:
    \begin{equation}
        x_j=\reportr{j},
        \qquad
        x_j<\reportr{j}\ \text{ and }\ \reportu{j}\ge \rho_e.
        \label{eq:WLE-X}
\end{equation}
That is, no individual user can make an additional trade directly with the pool that strictly improves its utility.

\end{enumerate}

\begin{fact}[IC implies no free lunch]
\label{fact:ic-no-free-lunch}
Under the reported-budget constraints, individual rationality, and pricing curve conformity, incentive compatibility implies
\emph{no free lunch}: for every report profile,
\[
    x_i>0\ \Longrightarrow\ y_i>0\quad(i\in\Yside),
    \qquad
    y_j>0\ \Longrightarrow\ x_j>0\quad(j\in\Xside).
\]
\end{fact}
\begin{proof}
Suppose a sell-$Y$ report of identity $i$ receives $x_i>0$ units of
$X$ while spending $y_i=0$. Fix all other reports and let
\[
    B:=\widehat y+\sum_{k\in\Yside\setminus\{i\}}\reportq{k}.
\]
Consider a true sell-$X$ user of type $(u,0)$, where $u$ will be
chosen below. If it submits the offending sell-$Y$ report, its net
gains are $\Delta_X=x_i>0$ and $\Delta_Y=0$. Neither
catastrophic-utility condition applies, so this deviation gives
utility $ux_i$.

Under truthful reporting, this user spends zero $X$ by its reported-budget
constraint, and its utility equals its $Y$ receipt. View the fixed
outside reports as truthful reports of distinct users. Individual
rationality then makes all receipts nonnegative. Writing $S_Y$ and
$R_Y$ for the total supply and receipt of $Y$ in this truthful
profile, pricing curve conformity requires $z=S_Y-R_Y>-\widehat y$.
Since all $Y$ supply comes from the outside sell-$Y$ orders,
\[
    R_Y=S_Y-z<S_Y+\widehat y\le B.
\]
The user's truthful receipt is at most $R_Y$, so its truthful
utility is bounded by $B$, independently of $u$. True values are
unbounded in the model, so choose a finite $u>B/x_i$. The deviation
then gives strictly greater utility, contradicting IC.

The sell-$X$ case is symmetric: if identity $j$ receives positive
$Y$ without spending $X$, use a true sell-$Y$ user of type $(v,0)$.
Its truthful $X$ receipt is bounded by
$x_0+\sum_{k\in\Xside\setminus\{j\}}\reportr{k}$, whereas its utility
from the offending report grows without bound as $v$ increases.
\end{proof}

\begin{fact}[Anonymity and IC imply economic symmetry]
\label{fact:ic-economic-symmetry}
For a deterministic anonymous AMM mechanism satisfying builder IC, 
identical reports in a profile receive equal utilities evaluated at their
common reported type.
\end{fact}
\begin{proof}
Consider two identical sell-$Y$ reports $(\reportv{},\reportq{})$,
with outcomes $(x_i,y_i)$ and $(x_j,y_j)$. Suppose
\[
    x_i-\reportv{}y_i<x_j-\reportv{}y_j.
\]
View order $i$ as the sole order of a builder-as-user, with true type $(\reportv{},\reportq{})$. 
Then the builder can swap the positions of orders $i$ and $j$.
Since their reports are identical, 
anonymity preserves the outcomes assigned to the two positions,
while the builder receives the outcome previously assigned to order $j$.
The common reported budget ensures that $0\le y_j\le\reportq{}$, so this
trade is feasible for the builder's true type. Its utility increases
from $x_i-\reportv{}y_i$ to $x_j-\reportv{}y_j$, contradicting builder-IC.

Reversing the roles of $i$ and $j$ rules out the opposite strict
inequality. Hence $x_i-\reportv{}y_i=x_j-\reportv{}y_j$.
For two identical sell-$X$ reports $(\reportu{},\reportr{})$, the same
argument gives $y_i-\reportu{}x_i=y_j-\reportu{}x_j$.
\end{proof}

\section{Impossibility}
\label{sec:no-censorship-impossibility}

In this section, we prove our main result: it is not possible to have a mechanism that satisfies all desired properties. 
Specifically, as in the model, we require individual rationality, pricing curve conformity, and anonymity. We first restate the theorem.

\begin{theorem}[IC + wLE $\Longrightarrow$ impossible]
\label{thm:no-censorship-impossibility}
No deterministic AMM mechanism can simultaneously satisfy incentive compatibility and weak local
efficiency. 
This impossibility holds as long as a strategic player is allowed to 
misreport its valuation and budget and inject a finite number of sybil bids and the builder is able to reordering. 
That is, the impossibility holds even when the consensus layer provides 
censorship resilience and block space is unlimited.
\end{theorem}

\paragraph{Proof sketch.} 
Consider a profile including both Sell-$Y$ and Sell-$X$ orders. 
We first establish a key lemma showing that eligible orders on the minority side 
with the smaller aggregate reported budget, measured in common units at the initial exchange rate, 
must trade their full budgets at that price (\cref{lem:nc-short-side}). This fixes the outcomes of orders on the minority side. 
Although the statement of this lemma is seemingly intuitive, its proof is highly nontrivial. The weakness of the wLE assumption requires us to
build the argument incrementally through a sequence of
carefully chosen intermediate claims.

We then combine this observation with Sybil identities and budget over-reporting. 
Suppose Sell-$X$ is the minority side. A strategic Sell-$Y$ user can inflate its reported budget 
and simultaneously submit an opposite-direction order to \emph{buy back} the excess $Y$ at the initial price. 
This enables the user to present more trading demand while keeping its net sale within its true budget.

In the auxiliary construction, the user splits its inflated budget evenly between two Sell-$Y$ bids, 
each matching another user's zero-value report. This creates three identical Sell-$Y$ orders. 
Economic symmetry, established in Fact~\ref{fact:ic-economic-symmetry}, forces these orders to receive equal amounts of $X$, 
thereby bounding the strategic user's share of the total trading loss. 
Here, trading loss is the shortfall in $X$ received relative to exchanging the same amount of $Y$ at the initial price.

For the final contradiction, we construct a profile with two identical Sell-$Y$ orders and a small Sell-$X$ order. 
The loss bound, together with value and budget truthfulness, yields a large utility guarantee for either Sell-$Y$ user 
through a deviation that reports value zero. Incentive compatibility then requires the utility guarantee to hold for both Sell-$Y$ users
under truthful reporting in the same profile. Their sum, however, exceeds the total utility available 
from the opposite-side order and the pool, yielding a contradiction.

\subsection{Notations and Useful Observations}
Recall that $z:=S_Y-R_Y$ represents the net amount of $Y$ sent to the pool. 
We use $Y(\cdot,\cdot)$ and $X(\cdot,\cdot)$ to denote a sell-$Y$ order and a sell-$X$ order, respectively. 
The arguments of $Y(\cdot,\cdot)$ and $X(\cdot,\cdot)$ specify reported parameters and need not coincide with
the submitting user's true type. When invoking incentive compatibility
or individual rationality, we specify the user's true type separately.

\paragraph{Normalizing the initial price.} 
We first change units so that the initial marginal exchange rate is $1$. 
Let $\sigma_0=F'(0)>0$ be the original initial price, measured in units of $X$ per unit of $Y$. Leave $Y$-amounts unchanged and define one new unit of $X$ to represent $\sigma_0$ original units. 
Thus an original $X$-amount $x$ becomes $\widehat x=x/\sigma_0$.

Budgets and values must be expressed in the same new units. 
For a sell-$Y$ order $Y(v,q)$ and a sell-$X$ order $X(u,r)$, the normalized parameters are
$$
\widehat v=\frac{v}{\sigma_0},\qquad
\widehat q=q,\qquad
\widehat u=u\sigma_0,\qquad
\widehat r=\frac{r}{\sigma_0}.
$$
The two values $v$ and $u$ scale differently because $v$ is measured in $X$ per $Y$, whereas $u$ is measured in $Y$ per $X$. 
In particular, one new unit of $X$ contains $\sigma_0$ original units, so its value in $Y$ is $u\sigma_0$.

The pool's $Y$-input is unchanged, while its $X$-output is divided by $\sigma_0$. The normalized pricing curve is therefore
$$
\widehat F(z):=\frac{F(z)}{\sigma_0},
\qquad
\widehat F'(0)=1.
$$

This rescaling preserves utility comparisons. For a sell-$Y$ user, we have
$$
\widehat x-\widehat v y=\frac{x-vy}{\sigma_0}.
$$
Thus a sell-$Y$ user's utility is multiplied by a fixed positive
factor under rescaling, so the ranking of truthful reporting and deviations is unchanged.
For a sell-$X$ user, the normalized utility remains the same:
$$
y-\widehat u\,\widehat x=y-ux.
$$
The same calculation applies to a user's aggregate net trade when it submits multiple orders. 
Moreover, budget constraints and trade directions are also preserved. Thus the change of units preserves UIC, individual-rationality, pricing curve conformity, and the efficiency requirements.

The original upper bound $x_0$ on $X$-payouts becomes $K:=x_0/\sigma_0>0$. From now on, we omit the hats and use the original symbols for the normalized quantities. We then have
$$
F'(0)=1,\qquad F(z)<K,\qquad F(z)\le z.
$$
where the last inequality follows because $F$ is concave: $F(z)\le F(0)+F'(0)z=z$.

The reverse curve remains $H(x)=-F^{-1}(-x)$, and now satisfies $H'(0)=1$. 
Consequently, in either trading direction, an order is \textit{strictly eligible} precisely when its
reported value is below $1$.

\paragraph{Consequences of weak local efficiency.}
Since the pool's marginal price is always positive in both directions, a zero-value order with unused budget could strictly improve its utility by making a sufficiently small additional trade with the pool. 
Thus wLE requires every zero-value order $Y(0,q)$ or $X(0,r)$ to spend its full reported budget.

Moreover, concavity and $F'(0)=1$ imply that the marginal price is at least $1$ for selling $X$ when $z\ge0$, and for selling $Y$ when $z\le0$. 
Thus, wLE gives
\begin{equation}
\begin{aligned}
    z\ge0&\quad\Longrightarrow\quad
        \text{every strictly eligible sell-$X$ order spends its full budget},\\
    z\le0&\quad\Longrightarrow\quad
        \text{every strictly eligible sell-$Y$ order spends its full budget}.
\end{aligned}
\label{eq:nc-sign-wle}
\end{equation}
Otherwise, with unused budget, a strictly eligible order whose reported value is below $1$ could profit from a sufficiently small additional trade with the pool.

\begin{lemma}[Outcome equivalence to zero-value report]
\label{lem:nc-taxation}
Fix all reports other than one sell-$Y$ order. For $s\ge 0$, let $g(s)$ be the $X$
received by this order when it reports $Y(0,s)$. If reporting $Y(v,q)$
produces outcome $(x,y)$, then
\begin{equation}
    x=g(y),\qquad
    x-vy\ge g(s)-vs\quad\text{for every }0\le s\le q.
    \label{eq:nc-realized-budget}
\end{equation}
The same statement holds with $X$ and $Y$ exchanged.
\end{lemma}
\begin{proof}
By wLE, every zero-value report $Y(0,s)$ spends its full reported budget $s$.

A sell-$Y$ user of true type $(0,y)$ can report $Y(v,q)$ without exceeding its true budget,
since the realized sale is exactly $y$. 
Incentive compatibility therefore requires $g(y)\ge x$.

Conversely, a sell-$Y$ user of true type $(v,q)$ can report $Y(0,y)$ and spend the same amount $y$.
Incentive compatibility again gives $x \ge g(y)$.
Together, these inequalities imply $x=g(y)$.

Finally, a sell-$Y$ user of true type $(v,q)$ can also report $Y(0,s)$ for any $s\le q$ and obtain utility $g(s)-vs$. Incentive compatibility
therefore gives the second inequality.
\end{proof}

\subsection{Trading at the Initial Price}
We now show that the minority side with smaller total reported budget must trade fully at the initial
price. We begin with two opposite orders and then extend the argument to arbitrary report profiles.

\begin{lemma}[Compensation against one opposite order]
\label{lem:nc-pair}
Against a single order $Y(v,q)$ with $0\le v<1$, order $X(0,r)$ receives
\begin{equation}
    \Gamma_q^X(r)=
    \begin{cases}
        r,&0\le r\le q,\\
        q+H(r-q),&r\ge q.
    \end{cases}
    \label{eq:nc-pair-X}
\end{equation}
Similarly, against a single order $X(u,r)$ with $0\le u<1$, order $Y(0,q)$ receives
\begin{equation}
    \Gamma_r^Y(q)=
    \begin{cases}
        q,&0\le q\le r,\\
        r+F(q-r),&q\ge r.
    \end{cases}
    \label{eq:nc-pair-Y}
\end{equation}
\end{lemma}
\begin{proof}
Throughout the proof, unless otherwise stated, let $(x_1,y_1)$ and $(x_2,y_2)$ denote the outcomes
of the sell-$Y$ and sell-$X$ orders, respectively, for the report profile
under consideration. Thus the sell-$Y$ order spends $y_1$ and receives
$x_1$, while the sell-$X$ order spends $x_2$ and receives $y_2$.

We first establish the result for two zero-value orders with equal
budgets, then extend it to arbitrary budgets, and finally allow
the opposite order to have any reported value below $1$.

\emph{Step 1: Zero values and equal budgets.}
Consider $Y(0,q),X(0,q)$ with $q\ge0$. If $q=0$, both orders spend
zero, and individual rationality gives $x_1,y_2\ge0$. Pricing curve
conformity gives $x_1=F(-y_2)$. Since $F$ is strictly increasing and
$F(0)=0$, both receipts must be zero. Assume henceforth that $q>0$.
By wLE, both orders spend their full budgets, so
$y_1=x_2=q$. 

Suppose $y_2<q$. For $0\le u<1$, let $b(u)$ be the receipt
of $X(u,q)$ against the fixed order $Y(0,q)$. A sell-$X$ user of true
type $(0,q)$ can instead report
$X(u,q)$, so incentive compatibility gives $b(u)\le y_2$.
Since the sell-$Y$ order still spends $q$, the net pool input is $z=q-b(u)\ge q-y_2>0$. 
Then weak local efficiency forces the sell-$X$ order to spend $q$ as well. 
Conversely, a sell-$X$ user of true type $(u,q)$ can report $X(0,q)$,
so incentive compatibility gives
$$
    b(u)-uq\ge y_2-uq.
$$
Thus $b(u)=y_2$ for every $0\le u<1$. Individual rationality
then requires $y_2\ge uq$ for every such $u$. Letting
$u\uparrow1$ gives $y_2\ge q$, a contradiction.

Symmetrically, we also have $x_1 \ge q$. If $y_2>q$, then
$z=q-y_2<0$ and the pricing curve conformity would give $x_1=q+F(z)<q$, another contradiction.
Hence $y_2=q$, and conformity gives $x_1=q$ as well.

\emph{Step 2: Zero values and unequal budgets.}
Consider $Y(0,q),X(0,r)$ with $0\le r\le q$. 
By weak local efficiency, we have $y_1=q$ and $x_2=r$.
A sell-$X$ user of true type $(0,q)$ can report $X(0,r)$,
so the equal-budget result and incentive compatibility gives $y_2\le q$.

For $0\le u<1$, let $b_r(u)$ be the receipt of $X(u,r)$ against the same order $Y(0,q)$. 
Incentive compatibility for a sell-$X$ user of true type $(0,r)$ gives $b_r(u)\le y_2\le q$.
Thus, the net pool input under the case of $Y(0,q),X(u,r)$ is nonnegative,
so wLE forces $X(u,r)$ to spend $r$. 
The reverse incentive comparison, in which a sell-$X$ user of
true type $(u,r)$ reports $X(0,r)$, gives
\[
    b_r(u)-ur\ge y_2-ur.
\]
Consequently $b_r(u)=y_2$ for every $0\le u<1$. Individual
rationality gives $y_2 \ge ur$ for every such $u$, hence $y_2 \ge r$.

Applying this argument at $r=q$, together with the
equal-budget result, shows that $X(u,q)$ spends and receives
$q$ for every $0\le u<1$. A sell-$X$ user of true type $(u,q)$ can
under-report its budget by submitting $X(u,r)$, so
$$
    q-uq\ge y_2-ur,
    \qquad\text{or equivalently}\qquad
    y_2\le r+(1-u)(q-r).
$$
Letting $u\uparrow1$ gives $y_2\le r$. Thus $y_2=r$, and
pricing curve conformity gives $x_1=r+F(q-r)$. 

Exchanging the trading directions proves the other case, $q\le r$, using the reverse curve $H$. 
This characterizes the outcomes when both reported values are zero.

\emph{Step 3: An arbitrary strictly eligible opposite order.}
It remains to allow $Y(v,q)$ against $X(0,r)$, where $0\le v<1$. 
The zero-value sell-$X$ order spends $x_2=r$.
By Lemma~\ref{lem:nc-taxation}, replacing the order $Y(v,q)$ by
$Y(0,y_1)$ gives this order the same outcome $(x_1,y_1)$.
According to the zero-value case just proved, the sell-$Y$ order's receipt is
$x_1=\Gamma_r^Y(y_1)$. 
Indeed, a sell-$Y$ user of true type $(v,q)$ can report $Y(0,t)$ for any $t\in[0,q]$,
and get an outcome of selling $t$ units of $Y$ and getting $\Gamma_r^Y(t)$ units of $X$. 
Incentive compatibility requires that the sell-$Y$ order's realized sale $y_1$ 
maximizes $\Gamma_r^Y(t)-vt$ over $t\in[0,q]$. 
For $t \le \min\{q,r\}$, $\Gamma_r^Y(t) = t$ and this utility expression equals $(1-v)t$, strictly increasing with $t$.

If $q\le r$, the order therefore spends and receives $y_1=x_1=q$.
Conformity then gives $y_2=q+H(r-q)$.
If $q\ge r$, no maximizer can lie below $r$, so $y_1\ge r$. 
Recall the zero-value sell-$X$ order spends $x_2=r$.
The sell-$Y$ receipt is therefore $x_1=r+F(y_1-r)$. 
Conformity and strict increase of $F$ further implies $y_2=r$. 
This proves \eqref{eq:nc-pair-X} for all $v<1$. 
Exchanging directions proves \eqref{eq:nc-pair-Y} for all $u<1$.
\end{proof}

\begin{lemma}[The smaller-budget side trades at the initial price]
\label{lem:nc-short-side}
Suppose every reported value is strictly below $1$. Let $B_Y$ and $B_X$
denote the total reported budgets of the sell-$Y$ and sell-$X$ orders,
respectively:
\[
    B_Y:=\sum_{i\in\Yside}\reportq{i},\qquad
    B_X:=\sum_{j\in\Xside}\reportr{j}.
\]
If $B_Y\ge B_X$, the net pool input satisfies $z\ge0$, and every sell-$X$
order has outcome $x_j=y_j=\reportr{j}$. If $B_Y\le B_X$, then $z\le0$
and every sell-$Y$ order has outcome $x_i=y_i=\reportq{i}$.
In particular, if $B_Y=B_X$, all orders trade their full budgets at the
initial price, and the pool is unchanged.
\end{lemma}
\begin{proof}
We first determine the sign of the net pool input, 
then treat balanced profiles, and finally characterize the outcome of 
each order on the smaller-budget side.

\emph{Step 1: The sign of the pool trade agrees with the budget imbalance.}
Suppose $B_Y\ge B_X$ but $z<0$. By wLE, every sell-$Y$
order spends its full budget. Fix any such order $i$, 
with report $Y(\reportv{i}, \reportq{i})$ and outcome $(x_i,\reportq{i})$.

Consider all other orders as one user's Sybil bids. Their net
$X$ sale and net $Y$ receipt are, respectively,
\[
    c_i := x_i-F(z)>0,\qquad \reportq{i}-z>\reportq{i}.
\]
A true zero-value sell-$X$ user with budget $c_i$, facing only order
$i$, can submit these orders and obtain that net outcome. 
Incentive compatibility requires its truthful receipt to be at least $\reportq{i}-z$. 
By Lemma~\ref{lem:nc-pair}, that truthful receipt is $\Gamma_{\reportq{i}}^X(c_i)$, so
\[
    \reportq{i}-z\le\Gamma_{\reportq{i}}^X(c_i).
\]
If $c_i\le\reportq{i}$, this gives $\reportq{i}-z\le c_i\le\reportq{i}$, contradicting $z<0$.
Thus $c_i>\reportq{i}$, and the inequality becomes $H(c_i-\reportq{i})\ge-z = H(-F(z))$.
Since $H$ is strictly increasing, $c_i-\reportq{i}\ge-F(z)$.
Substituting $c_i=x_i-F(z)$ gives $x_i\ge\reportq{i}$.
Summing over all sell-$Y$ orders gives
$R_X\ge B_Y\ge B_X\ge S_X$, contradicting $R_X-S_X=F(z)<0$.
Hence $B_Y\ge B_X$ implies $z\ge0$. Exchanging the trading directions proves that
$B_Y\le B_X$ implies $z\le0$.

\emph{Step 2: Consider a balanced profile.}
Suppose $B_Y=B_X$. The preceding result gives $z=0$. Weak local efficiency
forces every order to spend its full reported budget. 

Fix a sell-$Y$ order $i$ with spending $\reportq{i}$ and receipt $x_i$. 
All other orders together now have net $X$ sale $x_i$ and net $Y$ receipt $\reportq{i}$.
The same Sybil comparison gives $\reportq{i}\le\Gamma_{\reportq{i}}^X(x_i)$,
which implies $x_i\ge\reportq{i}$. 

On the other hand, pricing curve conformity and full spending give
$R_X=S_X=B_X=B_Y$. The receipts $x_i$ therefore sum to the total reported budgets, 
so every inequality $x_i\ge\reportq{i}$ must be an equality. 
The argument for sell-$X$ orders gives $y_j=\reportr{j}$. Thus every order trades its
full budget at the initial price.

\emph{Step 3: Consider one order on the smaller-budget side.}
Suppose $B_Y\ge B_X$. Fix a sell-$X$ order $j$ and all other reports.
Let $B_{X,-j}$ be the other sell-$X$ orders' total reported budget and
set $r^*=B_Y-B_{X,-j}$. Then the order $j$'s reported budget is at most $r^*$.
For every report $X(u,s)$ with $0\le u<1$ and $0\le s\le r^*$, the sign
result and wLE force this order to spend $s$.
For a fixed $s$, let $b_s(u)$ denote its receipt. Incentive
compatibility between any two values $u,u'\in[0,1)$ gives
\[
    b_s(u)-us\ge b_s(u')-us,
    \qquad
    b_s(u')-u's\ge b_s(u)-u's.
\]
Thus $b_s(u)=b_s(u')$: the receipt is independent of the
reported value. Denote this common receipt by $\beta(s)$.
Individual rationality gives $\beta(s)\ge us$ for every
$0\le u<1$. Letting $u\uparrow1$ yields $\beta(s)\ge s$.

At budget $r^*$ the profile is balanced, giving $\beta(r^*)=r^*$.
A true sell-$X$ user of type $(u,r^*)$ can submit $X(u,s)$, so
\[
    (1-u)r^*\ge\beta(s)-us,
    \qquad\text{hence}\qquad
    \beta(s)\le s+(1-u)(r^*-s).
\]
As this holds for every $u<1$, letting $u\uparrow1$ gives $\beta(s)\le s$.
Therefore $\beta(s)=s$ for every $0\le s\le r^*$. Taking
$s=\reportr{j}$ proves $x_j=y_j=\reportr{j}$. Exchanging the trading
directions proves the sell-$Y$ statement when $B_Y\le B_X$.
\end{proof}

\begin{remark}
The recent work of~\cite[Lemma~17 and Remark~19]{otter} has a similar result: under a deterministic \emph{natural} mechanism satisfying IC and wLE, every strictly eligible order on the minority side spends its full reported budget and is compensated at the initial spot price. The two results, however, fundamentally do not subsume each other.

As discussed in the introduction, the model of~\cite{otter} allows part of the pool's output to be redistributed outside the batch, thereby admitting a broader class of mechanisms than ours. Their proof, however, relies on additional assumptions that ours does not require.

In particular, in \cite{otter}
\begin{itemize}
    \item the first assumption is called \emph{no price reversal}
    \textsf{(NPR)}, which implies
    $$
        B_Y\ge B_X \implies z\ge0,
        \qquad
        B_Y\le B_X \implies z\le0.
    $$
    We instead \emph{directly prove} this sign property for strictly eligible profiles, when tokens output by the pool to be completely distributed to the users. Again as aforementioned, this proof is highly nontrivial, and we have to build the argument incrementally through a sequence of
    carefully chosen intermediate claims.
    \item the second assumption is called \textsf{(Drop)}, which assigns zero allocation and
    transfer to ineligible orders. Their proof uses this condition
    in incentive comparisons between eligible and ineligible
    reports. Our lemma, by contrast, concerns profiles consisting
    entirely of strictly eligible orders, and our proof never
    introduces ineligible reports.
\end{itemize}
Fairly speaking, the most significant contribution of~\cite{otter} is the
construction of a new mechanism enabled by relaxing exact
settlement. Their Lemma~17 serves primarily to justify the
formulation of social welfare maximization in a two-sided trading
setting. In our paper, by contrast, the characterization is itself a
central technical contribution: its proof makes essential use of IC in several distinct yet interdependent ways, linking strategic deviations across profiles to derive the structural constraints that drive our impossibility result.

\end{remark}

\subsection{The Sybil Deviation and the Contradiction}

For $a\ge0$, define
\[
    C(a):=a-F(a)\ge0.
\]
This is the loss from sending $a$ units of $Y$ through the pool rather
than exchanging them at the initial price. 
Similarly, the individual loss of a sell-$Y$ order spending $y_i$ and receiving $x_i$ 
is defined as $y_i-x_i$.
The following lemma bounds how a total loss is distributed between two orders.

\begin{lemma}[Loss bounds from a Sybil buyback]
\label{lem:nc-buyback}
Consider the profile $Y(0,q),Y(0,s),X(0,R)$, where
\[
    s\le q\le2s,\qquad a=q+s-R\ge0.
\]
Let $c_1=q-x_1$ and $c_2=s-x_2$ be the losses of the two sell-$Y$
orders. Then
\begin{equation}
    c_1\le\frac23 C(a),\qquad c_2\ge\frac13 C(a).
    \label{eq:nc-buyback-cost}
\end{equation}
\end{lemma}
\begin{proof}
By wLE, both sell-$Y$ orders spend their full budgets. 
Since $q+s\ge R$, Lemma~\ref{lem:nc-short-side} implies that 
the sell-$X$ order spends and receives $R$. 
The pool input is therefore $a$, and conformity gives $c_1+c_2=q+s-\bigl(R+F(a)\bigr)
    =C(a)$.

Fix the outside orders $Y(0,s)$ and $X(0,R)$. 
A sell-$Y$ user of true type $(0,q)$ can replace its truthful order 
$Y(0,q)$ with two orders $Y(0,s)$ and one buyback order $X(0,2s-q)$. 
The buyback budget is nonnegative because $q\le2s$. 
The resulting profile is
$$
    \underbrace{Y(0,s), \quad Y(0,s), \quad X(0,2s-q)}
        _{\text{deviating user's orders}},
    \quad
    \underbrace{Y(0,s), \quad X(0,R)}
        _{\text{fixed orders}}.
$$
There are now three identical sell-$Y$ orders, with total budget $3s$,
and two sell-$X$ orders, with total budget $R+2s-q$.
Their budget difference remains $a$. 
By wLE, all sell-$Y$ orders spend $s$. 
By Lemma~\ref{lem:nc-short-side}, both sell-$X$ orders spend their full budgets 
and trade at the initial price. 
Thus the pool's net input is still $a$, and pricing curve conformity fixes the three sell-$Y$
orders' total receipt at
\[
    R+2s-q+F(a)=3s-C(a).
\]

The economic symmetry by Fact~\ref{fact:ic-economic-symmetry} gives equal utilities
to the three identical sell-$Y$ orders. Their reported value is zero,
so their receipts are equal. Each therefore receives $s-C(a)/3$.
The deviating user's two sell-$Y$ identities thus have combined receipt $2s-2C(a)/3$.
Its sell-$X$ identity spends and receives $2s-q$, giving
\begin{align*}
    \text{net $Y$ sale}&=2s-(2s-q)=q,\\
    \text{net $X$ receipt}&=2s-2C(a)/3-(2s-q)=q-2C(a)/3.
\end{align*}
The deviation respects the true net budget exactly. Incentive
compatibility therefore requires $x_1=q-c_1\ge q-2C(a)/3$, proving the first
bound. The second follows from $c_1+c_2=C(a)$.
\end{proof}

We are now ready to prove Theorem~\ref{thm:no-censorship-impossibility}.
\begin{proof}[Proof of Theorem~\ref{thm:no-censorship-impossibility}]
Consider the truthful profile $Y(v,q),Y(v,q),X(0,R)$, where
\[
    v:=\frac56,
    \qquad
    q:=20K,
    \qquad
    R:=30K,
    \qquad
    \varepsilon:=1-v=\frac16.
\]
We show that each sell-$Y$ user can obtain utility greater
than $8K/3$ by reporting value zero while the other users
remain truthful. We then show that their total truthful
utility cannot meet both guarantees.

\paragraph{Utility from reporting value zero.}
Suppose user $1$ reports $Y(0,q)$ while user $2$ reports
$Y(v,q)$ and the sell-$X$ user reports $X(0,R)$.
By wLE, user $1$ spends $q$. Since $2q\ge R$, 
Lemma~\ref{lem:nc-short-side} implies that the sell-$X$
order spends and receives $R$, and that the pool input
is nonnegative. If $y$ is user $2$'s spending, then
\[
    10K=R-q\le y\le q=20K.
\]
Let $a=q+y-R$ be the pool input, and define the individual losses
\[
    c_1=q-x_1,\qquad c_2=y-x_2.
\]

We can apply the bounds in Lemma~\ref{lem:nc-buyback} to these losses. To see why,
replace user $2$'s report by $Y(0,y)$.
Lemma~\ref{lem:nc-taxation} says its receipt $x_2$ remains the same.
In the new profile, user $1$ still spends $q$, user $2$ spends $y$,
and the sell-$X$ order still spends and receives $R$.
The pool input is therefore still $a$, so conformity fixes the sum
of the two sell-$Y$ receipts at $R+F(a)$. Since $x_2$ is unchanged,
$x_1$ is unchanged as well. The losses $c_1,c_2$ are consequently
the same as in the all-zero profile $Y(0,q),Y(0,y),X(0,R)$.
Because $y\le q\le2y$, Lemma~\ref{lem:nc-buyback} gives
\begin{equation}
    c_1\le\frac23 C(a),\qquad c_2\ge\frac13 C(a).
    \label{eq:nc-two-cost-bounds}
\end{equation}

User $2$ has another simple deviation: report $Y(0,R-q)$.
Together with $Y(0,q)$ and $X(0,R)$, this forms a balanced
profile. It therefore spends and receives $R-q$, obtaining true
utility $\varepsilon(R-q)$. Its current utility is
$x_2-vy=\varepsilon y-c_2$. Truthfulness requires
\[
    \varepsilon y-c_2\ge\varepsilon (R-q),
    \qquad\text{hence}\qquad c_2\le\varepsilon a.
\]
Combining this inequality with \eqref{eq:nc-two-cost-bounds} yields
\[
    C(a)\le3\varepsilon a=\frac a2.
\]
Thus $F(a)=a-C(a)\ge a/2$. Since the pool pays out less than $K$,
we obtain $a<2K$ and
\[
    c_1\le\frac23 C(a)\le\frac a3<\frac{2K}{3}.
\]

Evaluated at user $1$'s true value $v$, its utility from
reporting zero is
\begin{equation}
    x_1-vq=\varepsilon q-c_1>
    \frac{20K}{6}-\frac{2K}{3}=\frac{8K}{3}.
    \label{eq:nc-deviation-utility}
\end{equation}
Reporting value zero is within its true budget, because it spends
exactly $q$.

\paragraph{The two guarantees cannot both be met.}
Return to the truthful profile $Y(v,q),Y(v,q),X(0,R)$, 
and let $U_1,U_2$ be the two sell-$Y$ users' utilities.
Either sell-$Y$ user can report value zero while the other stays
truthful. Applying \eqref{eq:nc-deviation-utility} to each user and
using incentive compatibility gives
\begin{equation}
    U_1+U_2>\frac{16K}{3}=\frac{32K}{6}.
    \label{eq:nc-lower-bound}
\end{equation}

On the other hand, Lemma~\ref{lem:nc-short-side} implies
that the sell-$X$ order spends and receives $R$, and that
the pool's net $Y$ input $z$ is nonnegative. The two sell-$Y$ users
therefore spend $R+z$ and receive $R+F(z)$ in total. Their combined utility is
\[
    U_1+U_2=R+F(z)-v(R+z)=\varepsilon R+F(z)-vz.
\]
By concavity, $F(z)\le z$. Since $v\ge0$ and $F(z)<K$,
\[
    F(z)-vz\le(1-v)F(z)<\varepsilon K,
\]
and therefore
\begin{equation}
    U_1+U_2\le\varepsilon(R+K)=\frac{31K}{6}.
    \label{eq:nc-upper-bound}
\end{equation}
This contradicts \eqref{eq:nc-lower-bound} because $K>0$.
\end{proof}

\section{Feasibility Results}
\label{sec:tsf-feasibility}
The impossibility proof combines three strategic deviations: creating Sybil identities, 
over-reporting budgets, and misreporting values. 
In this section, we show that removing any one of the three capabilities restores feasibility. 
We consider three restricting settings -- no Sybil identities, no aggregate budget over-reporting, 
and public values -- and construct a mechanism for each setting that satisfies 
IC, wLE, and other natural properties including individual rationality, pricing curve conformity, and anonymity.

All three mechanisms share a common two-sided matching construction. 
For a profile containing both Sell-$Y$ and Sell-$X$ orders, 
we discard ineligible orders and compare the total reported budgets of eligible orders 
on the two sides, expressed in common units at the initial price. 
Eligible orders on the smaller-budget side spend their
full reported budgets at that price. The remaining problem is one-sided, on an augmented 
compensation curve that accounts for these opposite-direction trades.

We first describe this shared construction, which determines the outcomes of 
all orders on at least one side. We then present a one-sided mechanism for the remaining orders 
under each restriction. Each subsection specifies the admissible strategies 
and analyzes the resulting two-sided mechanism.

Throughout this section, $F,H,\sigma_0$, and $\rho_0$ retain their meanings
and units from the model. In particular, we do not require the initial
price to be normalized to one. We write $Y(\reportv{},\reportq{})$ and 
$X(\reportu{},\reportr{})$ for orders with the indicated reported parameters, 
and specify true types separately in incentive comparisons.

\subsection{A Shared Two-Sided Matching Construction}
\label{sec:tsf-common}

The mechanism first assigns outcome $(0,0)$ to every initially
ineligible order whose ask exceeds the initial price in its trading direction. 
Let $\Yside_0$ and $\Xside_0$ denote the sets of eligible sell-$Y$
orders and eligible sell-$X$ orders, respectively:
\[
    \Yside_0:=\{i\in\Yside:\reportv{i}\le\sigma_0\},
    \qquad
    \Xside_0:=\{j\in\Xside:\reportu{j}\le\rho_0\}.
\]
Define their total reported budgets by
\[
    B_Y:=\sum_{i\in\Yside_0}\reportq{i},
    \qquad
    B_X:=\sum_{j\in\Xside_0}\reportr{j}.
\]
At the initial price, $B_X$ units of $X$ exchange for
$\rho_0 B_X$ units of $Y$. We call a profile \emph{Sell-$Y$ dominated}
if $B_Y\ge\rho_0 B_X$, and \emph{Sell-$X$ dominated} otherwise.
Balanced profiles are classified as Sell-$Y$ dominated.

\paragraph{The augmented compensation curves.}
For a Sell-$Y$-dominated profile, set $M:=\rho_0B_X$. 
First, assign each eligible sell-$X$ order its full-budget outcome at the initial price: 
\[
    x_j=\reportr{j},\qquad y_j=\rho_0\reportr{j}.
\]
Together, these orders supply $B_X=\sigma_0M$ units of $X$ and require $M$ units
of $Y$. The eligible sell-$Y$ orders therefore face the augmented curve
\begin{equation*}
    G_M^Y(t):=
    \begin{cases}
        \sigma_0t,&0\le t\le M,\\
        \sigma_0M+F(t-M),&t\ge M.
    \end{cases}
\end{equation*}
The first $M$ units of $Y$ are exchanged at the initial price 
with the $X$ tokens supplied by the sell-$X$ orders, 
and any additional $Y$ is sent to the pool. We determine the outcomes of 
the eligible Sell-$Y$ orders in $\Yside_0$ by applying a one-sided mechanism to $G_M^Y$.

Symmetrically, for a Sell-$X$-dominated profile, 
set $M:=\sigma_0B_Y$. Assign each eligible sell-$Y$ order the outcome $y_i=\reportq{i}$ and $x_i=\sigma_0\reportq{i}$. 
We then apply the corresponding one-sided mechanism to 
the orders in $\Xside_0$ using the augmented compensation curve
\begin{equation*}
    G_M^X(t):=
    \begin{cases}
        \rho_0t,&0\le t\le M,\\
        \rho_0M+H(t-M),&t\ge M.
    \end{cases}
\end{equation*}

The two pieces of each curve meet with the same marginal slope. Thus
both augmented curves are continuous, strictly increasing, and concave,
with finite positive initial marginal compensation and bounded output.

Here, $t$ denotes the cumulative input to the one-sided mechanism, 
measured in $Y$ for $G_M^Y$ and in $X$ for $G_M^X$, 
including the amount exchanged with opposite-side orders. It is different from
the net $Y$ input $z$ to the pool. For example, in the Sell-$Y$-dominated
case, the one-sided mechanisms introduced below ensure that 
the sell-$Y$ orders spend a total of $P\ge M$, so the 
assigned sell-$X$ outcomes can be settled. The remaining
$Y$ is sent to the pool, giving $z=P-M$.

\paragraph{Requirements on the one-sided mechanism.}
The matching construction reduces the remaining problem to
one side, but it does not specify how orders on that side share the
available compensation. The three mechanisms below make this choice
differently. We first state the properties that each one-sided rule
must satisfy to work with the matching construction.

For the sell-$Y$ direction, let $G:[0,\infty)\to[0,\infty)$ be an
effective compensation curve with $G(0)=0$, and write
$\gamma(t):=G'_+(t)$ and $\gamma_0:=G'_+(0)$. We require $G$ to be
continuous, strictly increasing, and concave, with $0<\gamma_0<\infty$.
Here $G(t)$ is the total amount of $X$ available to compensate the
orders when they collectively spend $t$ units of $Y$.
The marginal compensation $\gamma(t)$ is the rate available for an
additional unit of $Y$. Concavity means that this rate cannot increase
as more $Y$ is traded. For the matching construction, $G$ is the
augmented curve $G_M^Y$.

The input orders have asks at most $\gamma_0$. Writing $P:=\sum_i y_i$,
the mechanisms below satisfy
\begin{equation}
\begin{gathered}
    0\le y_i\le\reportq{i},\qquad
    \reportv{i}y_i\le x_i\le\gamma_0y_i,
    \qquad \sum_i x_i=G(P),\\
    y_i<\reportq{i}\quad\Longrightarrow\quad
        \reportv{i}\ge G'_+(P).
\end{gathered}
\label{eq:tsf-kernel-requirements}
\end{equation}
These conditions have direct economic meanings. Each order spends no
more than its declared budget and receives compensation worth at least
its reported ask. Its compensation rate is also bounded above by the
initial rate, which in particular rules out receiving $X$ without
spending $Y$. The total compensation must equal $G(P)$, so the
allocation distributes exactly the output provided by the effective
curve. Finally, an order may have unspent budget only if its ask is
at least the ending marginal compensation: no further trade with
the pool would give it strictly positive surplus at its reported
value. This is the one-sided weak-local-efficiency requirement.
We also require anonymity. The mechanisms constructed below indeed use only
reported values and budgets and treat the input bids as a multiset,
so reordering has no effect on the execution outcome.

One further requirement ensures that the initial-price trades promised
to the opposite side can actually be settled. In a Sell-$Y$-dominated
profile, the eligible sell-$X$ orders have been assigned full-budget
trades and together must receive $M=\rho_0B_X$ units of $Y$.
These tokens must come from the sell-$Y$ orders. The one-sided mechanism
thus must ensure that their total spending satisfies $P\ge M$.

On the augmented curve $G_M^Y$, these first $M$ units of $Y$ are
exchanged at the constant initial rate $\sigma_0$. We state the
requirement in a more general form, which will be mainly used in~\autoref{sec:tsf-public-values} 
for the residual curves in the public-value mechanism. Let $L\ge0$ be any
input amount up to which the curve $G$ offers the constant initial
rate $\gamma_0$, meaning that $G(t)=\gamma_0t$ for all $0\le t\le L$.
The mechanism must execute at least $L$ units of input if enough
reported budget is available, and otherwise exhaust all reported
budgets:
\begin{equation}
    G(t)=\gamma_0t\ \text{for }0\le t\le L
    \quad\Longrightarrow\quad
    P\ge\min\left\{L,\sum_i\reportq{i}\right\}.
    \label{eq:tsf-linear-completion}
\end{equation}
We call this condition \emph{linear completion}. It does not require
the mechanism to stop at $L$; further trade may take place beyond it.
If the total reported budget is at most $L$, all orders spend their
full budgets and receive $x_i=\gamma_0\reportq{i}$.\footnote{The insufficient-budget case does not arise within this
Sell-$Y$-dominated branch. We include it because later incentive
proofs compare reports on a fixed augmented curve even when a
deviation changes which side is dominant. A same-direction budget
misreport leaves the sell-$X$ reports, and hence $M$, unchanged,
but may reduce $B_Y$ below $M$. The actual two-sided construction
then treats sell-$Y$ orders as the minority and gives them
full-budget trades at the initial price. Linear completion ensures
that an auxiliary execution on the unchanged curve $G_M^Y$
reproduces exactly these sell-$Y$ outcomes.}
For a Sell-$Y$-dominated profile, taking $G=G_M^Y$ and $L=M$
gives $P\ge M$, since $B_Y\ge M$.

Note that an order whose ask equals the initial rate earns zero surplus from
these trades, so leaving it unfilled need not violate weak local
efficiency. Linear completion requires the mechanism to execute such
zero-surplus trades as well, so that the promised opposite-side trades
can be settled.

The sell-$X$ requirements exchange the input and output tokens.
Each later subsection verifies these properties for its one-sided
rule; incentive compatibility requires a separate argument.

\paragraph{Ending prices at kinks.}
The ending prices $(\sigma_e,\rho_e)$ refer to the pool's final state
after settlement. When the pricing curve is differentiable there,
they are the marginal exchange rates in the two directions. At a
kink, however, the actual directional marginal compensations need
not be reciprocals. To specify the pair used in the
model's weak-local-efficiency conditions, we use the actual marginal
compensation in the dominant trading direction and define the other
price by reciprocity.

In the Sell-$Y$-dominated case, the pool receives a net amount
$S_Y-R_Y\ge0$ of $Y$. Using the sell-$Y$ compensation curve $F$, set
\[
    \sigma_e:=F'_+(S_Y-R_Y),
    \qquad \rho_e:=\frac{1}{\sigma_e}.
\]
In the Sell-$X$-dominated case, the pool receives a net amount
$S_X-R_X\ge0$ of $X$. Using the sell-$X$ compensation curve $H$, set
\[
    \rho_e:=H'_+(S_X-R_X),
    \qquad \sigma_e:=\frac{1}{\rho_e}.
\]
The right derivative in each case describes the compensation for an
additional trade in the dominant direction. These definitions give
a reciprocal pair of supporting prices at the ending pool state.
At a kink, the reciprocal price need not equal the actual marginal
compensation in the minority direction. Every eligible minority
order spends its full reported budget, so its weak-local-efficiency
condition imposes no further restriction on that price. When the
pool remains at its initial state, either definition gives
$(\sigma_e,\rho_e)=(\sigma_0,\rho_0)$ by (A3). 

\begin{lemma}[Settlement under the augmented curve]
\label{lem:tsf-settlement}
Suppose the one-sided mechanisms satisfy
\eqref{eq:tsf-kernel-requirements} and \eqref{eq:tsf-linear-completion},
and are anonymous. The two-sided matching
construction satisfies the reported-budget constraints, individual rationality, pricing curve conformity, 
anonymity, and weak local efficiency. Moreover, every finite collection $\playerids$ of identities
satisfies the portfolio price bound
\begin{equation}
    \Delta_X(\playerids)+\sigma_0\Delta_Y(\playerids)\le0.
    \label{eq:tsf-portfolio-price-bound}
\end{equation}
\end{lemma}
\begin{proof}
Suppose $B_Y\ge\rho_0B_X$. Since $B_Y\ge M$, linear completion gives
$P\ge M$. The majority orders supply enough $Y$ to pay the minority
orders. The remaining token flows satisfy
\[
    R_Y=M,\qquad S_X=B_X,\qquad S_Y=P,\qquad
    R_X=G_M^Y(P)=B_X+F(P-M).
\]
Consequently $R_X-S_X=F(S_Y-R_Y)$, as required. The other direction is
symmetric, using $H(t)=-F^{-1}(-t)$.

The reported-budget constraints and individual rationality follow
from \eqref{eq:tsf-kernel-requirements} and the minority outcomes.
For weak local efficiency, all eligible minority orders are fully spent.
An unfilled majority order has reported value at least the ending
marginal compensation of the pool in its trading direction, which is
the selected supporting price under the convention above.
The matching step uses only reported values and budgets, and the
one-sided mechanisms are anonymous. Their composition is therefore
anonymous.

Finally, every executed sell-$Y$ order satisfies $x_i\le\sigma_0y_i$,
and every executed sell-$X$ order satisfies $y_j\le\rho_0x_j$.
Summing these inequalities over any collection gives
\eqref{eq:tsf-portfolio-price-bound}.
\end{proof}

Lemma~\ref{lem:tsf-settlement} separates the analysis of each
feasibility result into two parts. First, we specify a one-sided
mechanism and verify the conditions stated above, 
including~\eqref{eq:tsf-kernel-requirements},~\eqref{eq:tsf-linear-completion}, and anonymity.
The lemma then guarantees that the resulting two-sided mechanism 
satisfies all desired properties except for IC, including the reported-budget constraints,
individual rationality, pricing curve conformity, weak local efficiency,
and anonymity. Second, we prove incentive compatibility of the
two-sided mechanism under the deviations allowed in the corresponding
setting. In particular, reordering has no effect on these mechanisms, so the same incentive
proofs apply to a builder-as-user under the same reporting restrictions. 
The following subsections carry out these two steps for each
of the three settings.

\subsection{No Sybil Identities: The Hand-in-Hand Mechanism}
\label{sec:tsf-no-sybil}

This subsection considers the setting in which each user has a
single fixed identity and may submit at most one order. It cannot
split its order or create additional identities to submit Sybil orders.
As before, the user may report an arbitrary value and budget, including
a budget above its true budget, and may choose either trading direction.

\paragraph{The one-sided mechanism.} 
We describe the one-sided mechanism for a Sell-$Y$-dominated profile. 
For a Sell-$X$-dominated profile, exchange the roles of $X$ and $Y$ 
and use the augmented curve $G_M^X$ instead. 

The eligible sell-$Y$ orders face the augmented curve 
$G:=G_M^Y$, with marginal compensation $\gamma(t):=G'_+(t)$. 
Each order starts with $x_i=y_i=0$ and remaining
budget $b_i=\reportq{i}$. For a ask $v\ge0$, define its
\emph{value cutoff} as the cumulative input beyond which the marginal
compensation falls below $v$:
\[
    \theta(v):=\sup\{t\ge0:\gamma(t)\ge v\},
\]
where the supremum of the empty set is zero. Here $t$ measures the
total $Y$ spent by all orders, not the spending of an individual order.
An order is willing to trade up to this cutoff, including any interval 
on which the marginal compensation equals the ask. 
At cumulative input $t$, an order is \emph{active} if it has remaining
budget and has not reached its value cutoff. Define the active set
\[
    A(t):=\{i:b_i(t)>0\text{ and }t<\theta(\reportv{i})\}.
\]
If $A(t)$ is empty, stop. Otherwise, all active orders trade each
marginal transaction in equal shares: for each $i\in A(t)$, set
\begin{equation}
    \mathrm dy_i=\frac{\mathrm dt}{|A(t)|},\qquad
    \mathrm dx_i=\frac{\mathrm dG(t)}{|A(t)|},\qquad
    \mathrm db_i=-\frac{\mathrm dt}{|A(t)|},
    \label{eq:tsf-equal-dynamics}
\end{equation}
while inactive orders keep their outcomes. Equivalently, with the active set
fixed, advance until some remaining budget is exhausted or some value
cutoff is reached, then remove the corresponding orders from the active set. This gives
finitely many stages.

\begin{theorem}[Feasibility without Sybil identities]
\label{thm:tsf-no-sybil}
The two-sided mechanism obtained by combining the initial-price matching construction
with the one-sided hand-in-hand mechanism satisfies incentive compatibility
when every user controls at most one identity. It also satisfies the properties in
Lemma~\ref{lem:tsf-settlement}.
\end{theorem}
\begin{proof}
Let $P$ be the terminal cumulative input. 
Summing \eqref{eq:tsf-equal-dynamics} gives $\sum_i y_i=P$ and
$\sum_i x_i=G(P)$. Each order trades only at marginal compensation
between its ask and $\gamma_0$, so
\eqref{eq:tsf-kernel-requirements} holds. An unfilled order has reached
its value cutoff, where $\gamma(P)\le\reportv{i}$; later trades can
only reduce the marginal compensation. The active sets and equal-share
updates depend only on reported asks and budgets, establishing anonymity.
On an initial linear segment, every
eligible order with remaining budget is willing to trade, so the
mechanism cannot stop before spending the smaller of the segment's
length and the total reported budget. This proves
\eqref{eq:tsf-linear-completion}. Lemma~\ref{lem:tsf-settlement}
therefore applies.

Fix a true sell-$Y$ user of type $(v,q)$ and all other reports.
First consider an arbitrary report $Y(\reportv{},\reportq{})$ in its
true direction. Let $M:=\rho_0B_X$, where $B_X$ is the total reported
budget of eligible sell-$X$ orders. These orders are unchanged by the
deviation, so $M$ remains fixed.
If sell-$Y$ remains dominated, the mechanism applies the one-sided rule to
$G_M^Y$. Otherwise, every eligible sell-$Y$ order receives a full-budget
trade at the initial price. Running the one-sided rule on $G_M^Y$
would give the same sell-$Y$ outcomes: the total available budget is
below $M$, so linear completion exhausts all reported budgets within
the curve's constant-price portion. Thus both the truthful and
deviating outcomes can be computed on the fixed curve $G_M^Y$, regardless
of which side dominates.

Compare the execution under reports $Y(v,q)$ and
$Y(\reportv{},\reportq{})$, at the same cumulative input $t$.
Under equal sharing, the user's reported value and budget determine only
when its order stops trading, not its share while active.
Until the user's order exits in one execution, the active sets evolve
in the same way: every other order has the same report and receives
the same allocations in both executions. Hence the user receives
exactly the same trades up to that first exit.

It remains to show that neither stopping earlier nor continuing
longer can help.
\emph{If the deviating order exits first,} the truthful execution
only adds trades that stay within the true budget and have nonnegative
true surplus.
\emph{If the truthful order exits first,} it has either exhausted
$q$ or reached its value cutoff $\theta(v)$. In the former case,
any additional deviating sale exceeds the true budget and gives
utility $-\infty$. In the latter case, any additional trade occurs
at marginal compensation at most $v$, so it cannot increase true
utility. Simultaneous exits give the same outcome, and an order
inactive from the start is treated as exiting at $t=0$.
Thus the deviation cannot increase utility.

It remains to consider a report in the opposite direction. A positive
receipt of $Y$ entails a positive sale of $X$, so the outcome is a
strict reversal of the user's desired trade and has utility $-\infty$.
An opposite report with no $Y$ receipt has nonpositive utility. Neither
outcome can improve on truthful reporting, which gives nonnegative
utility by individual rationality. Exchanging the two tokens proves
the claim for a true sell-$X$ user.
\end{proof}

\subsection{No Budget Over-Reporting: Proportional Hand-in-Hand}
\label{sec:tsf-no-overreport}

This subsection considers the setting in which each user may
create Sybil identities and report arbitrary values, but cannot
over-report its aggregate budget in its true trading direction.
Specifically, for a true sell-$Y$ user of type $(v,q)$, any finite collection $\playerids$ of identities must satisfy
\begin{equation}
    \sum_{i\in\playerids\cap\Yside}\reportq{i}\le q.
    \label{eq:tsf-budget-restriction}
\end{equation}
The user may under-report its budget or split it among multiple
sell-$Y$ identities with arbitrary reported values. It may also submit sell-$X$ identities with arbitrary
asks and declared budgets. The symmetric restriction applies
to a true sell-$X$ user. As in the original model, the true-budget
constraint on realized net sales remains unchanged.

\paragraph{The one-sided mechanism.}
For a Sell-$Y$-dominated profile, the eligible sell-$Y$ orders face the
augmented curve $G:=G_M^Y$. For a Sell-$X$-dominated profile, exchange
$X$ and $Y$ and use $G_M^X$.

Let $t$ denote the cumulative $Y$ input, including trades with
opposite-side orders. Start at $t=0$, with $x_i(0)=y_i(0)=0$ and
remaining budget $b_i(0)=\reportq{i}$. Determine the active set $A(t)$
using the same budget and value-cutoff conditions as in the hand-in-hand
mechanism. If $A(t)$ is empty, stop. Otherwise, let
$S(t):=\sum_{j\in A(t)}b_j(t)$ and divide each marginal transaction
in proportion to remaining budgets: for each $i\in A(t)$, set
\begin{equation*}
    \mathrm dy_i=\frac{b_i(t)}{S(t)}\,\mathrm dt,\qquad
    \mathrm dx_i=\frac{b_i(t)}{S(t)}\,\mathrm dG(t),\qquad
    \mathrm db_i=-\frac{b_i(t)}{S(t)}\,\mathrm dt.
\end{equation*}
Inactive orders retain their allocations and remaining budgets.
Update the active set as orders exhaust their budgets or reach their
value cutoffs, and continue until it is empty. At the terminal cumulative
input $P$, return the outcomes $(x_i(P),y_i(P))$. We call this the
\emph{proportional hand-in-hand mechanism}.

For any two active orders $i,j\in A(t)$, the budget update rule gives
\[
    \frac{\mathrm db_i}{b_i(t)}
        =\frac{\mathrm db_j}{b_j(t)}
        =-\frac{\mathrm dt}{S(t)}.
\]
Their remaining budgets therefore retain their initial ratio:
$b_i(t)/b_j(t)=\reportq{i}/\reportq{j}$.
Hence there is a common remaining-budget fraction $h(t)$ such that
\[
    b_i(t)=\reportq{i}h(t)\qquad(i\in A(t)),
\]
with $h(0)=1$.

To obtain an explicit expression for $h(t)$, let
\[
    W(t):=\sum_{i:\,t<\theta(\reportv{i})}\reportq{i}
\]
be the total initial declared budget of orders whose value cutoffs
have not been reached. Before the process stops, their total remaining
budget is $S(t)=h(t)W(t)$. Substituting this and
$b_i(t)=\reportq{i}h(t)$ into the budget update rule gives
$h'(t)=-1/W(t)$. Thus
\[
    h(t)=1-\int_0^t\frac{\mathrm ds}{W(s)},\qquad
    \frac{b_i(t)}{S(t)}=\frac{\reportq{i}}{W(t)}
    \quad(i\in A(t)).
\]
Thus, each marginal transaction is divided among active orders in
proportion to their initial declared budgets. An order exits at its value
cutoff, or all remaining active budgets are exhausted together when $h$
reaches zero. The terminal cumulative input satisfies
$P=\sum_i y_i\le\sum_i\reportq{i}$, since each order spends at most
its declared budget.

\begin{lemma}[Same-direction incentive compatibility]
\label{lem:tsf-proportional-incentives}
On a fixed curve $G$, proportional hand-in-hand is user incentive
compatible against same-direction deviations satisfying
\eqref{eq:tsf-budget-restriction}.
\end{lemma}
\begin{proof}
Fix $G$, a true sell-$Y$ user of type $(v,q)$, and all other reports.
Let $\playerids$ denote the identities the user submits in a
same-direction deviation. The case $q=0$ is immediate, so assume
$q>0$ and omit zero-budget orders. In the deviating execution,
let $W_{-\playerids}(t)$ and $W_{\playerids}(t)$ denote the contributions
to $W(t)$ from the other users' orders and from the deviating user's
identities, respectively. Thus
\[
    W(t)=W_{-\playerids}(t)+W_{\playerids}(t).
\]
At each position reached before the deviating execution stops that has
positive true marginal surplus, the truthful report contributes $q$
to $W(t)$, whereas $W_{\playerids}(t)\le q$. Consequently, the user's
share under the deviation is at most its truthful share:
\[
    \frac{W_{\playerids}(t)}{W_{-\playerids}(t)+W_{\playerids}(t)}
       \le\frac{q}{W_{-\playerids}(t)+q}.
\]
It remains to show that the truthful execution reaches every such
position. By the formula for $h(t)$, remaining active budgets are
exhausted when $\int_0^t W(s)^{-1}\,\mathrm ds$ reaches one. Before
the deviating execution stops, at any position $t$ with positive
true marginal surplus,
\[
    \int_0^t\frac{\mathrm ds}{W_{-\playerids}(s)+q}
       \le\int_0^t\frac{\mathrm ds}
           {W_{-\playerids}(s)+W_{\playerids}(s)}<1.
\]
The truthful order therefore still has positive remaining budget,
and its value cutoff has not been reached. Hence the truthful
execution reaches $t$. Integrating the nonnegative true marginal
surplus proves that truthful utility is at least the utility from the
deviation's profitable trades. Any remaining deviating trades have
nonpositive true surplus. The sell-$X$ case is symmetric.
\end{proof}

The preceding lemma treats same-direction deviations on a fixed curve.
Opposite-direction identities can change the initial-price matching
capacity $M$. The next lemma provides the utility bounds needed to
handle such deviations.

\begin{lemma}[Utility bounds for initial-price matching]
\label{lem:tsf-proportional-utility-bounds}
Fix a true sell-$Y$ user of type $(v,q)$ with $q>0$ and
$0\le v<\sigma_0$, and let $\playerids$ denote the identities it
uses in a deviation. Fix all other reports, and let $B_{-\playerids}$
be the total declared budget of the other eligible sell-$Y$ orders.
Let $V(M;v,q)$ be the utility of the truthful report $Y(v,q)$ on
$G_M^Y$. For every $d\ge0$,
\begin{align}
    V(M;v,q)&\ge(\sigma_0-v)
        \min\{q,(M-B_{-\playerids})_+\},
        \label{eq:tsf-proportional-prefix-lower}\\
    V(M+d;v,q)&\le V(M;v,q)+(\sigma_0-v)d.
        \label{eq:tsf-proportional-prefix-bound}
\end{align}
Here $(\cdot)_+$ denotes the maximum of its argument and zero.
\end{lemma}
\begin{proof}
For the lower bound, all eligible orders remain willing throughout
the initial-price segment. The user's share is therefore
$q/(B_{-\playerids}+q)$, and its input on this segment is
$q\min\{M/(B_{-\playerids}+q),1\}$. Each such unit earns surplus
$\sigma_0-v$, and later truthful trades have nonnegative surplus. Thus
\[
    V(M;v,q)\ge(\sigma_0-v)q\min\{M/(B_{-\playerids}+q),1\}
      \ge(\sigma_0-v)\min\{q,(M-B_{-\playerids})_+\}.
\]

For the upper bound, consider a virtual execution in which the user
remains willing until its budget is exhausted. Parameterize this
execution by the fraction $\lambda=1-h(t)\in[0,1]$ of its budget
spent, and let $\pi_M(\lambda)$ be the marginal compensation on
$G_M^Y$ at that point. Since the user's remaining budget is $qh(t)$,
each increment $\mathrm d\lambda$ allocates $q\,\mathrm d\lambda$
units of input to it. The truthful and virtual executions agree up
to the user's value cutoff or budget exhaustion. Marginal compensation
is nonincreasing, so its truthful utility is
\begin{equation}
    V(M;v,q)=q\int_0^1(\pi_M(\lambda)-v)_+
        \,\mathrm d\lambda.
    \label{eq:tsf-proportional-clock}
\end{equation}
The positive part discards the virtual trades below the user's true value.

Write $\delta:=d/(B_{-\playerids}+q)$. If $\delta<1$, the extra
$d$ units of initial-price input on $G_{M+d}^Y$ consume the fraction
$\delta$ of each eligible order's budget. Thereafter, positions
$d+t$ on $G_{M+d}^Y$ and $t$ on $G_M^Y$ have the same marginal
compensation and the same other orders whose value cutoffs have not
been reached. Thus their totals $W$ agree at corresponding positions,
and $\mathrm d\lambda=\mathrm dt/W(t)$ gives the same subsequent
increments of spent-budget fraction. Consequently,
\[
    \pi_{M+d}(\lambda)=
    \begin{cases}
        \sigma_0,&0\le\lambda<\delta,\\
        \pi_M(\lambda-\delta),&\delta\le\lambda<1.
    \end{cases}
\]
Equation~\eqref{eq:tsf-proportional-clock} now gives
\begin{align*}
    V(M+d;v,q)
       &=q\delta(\sigma_0-v)
         +q\int_0^{1-\delta}(\pi_M(\lambda)-v)_+
             \,\mathrm d\lambda\\
       &\le V(M;v,q)+q\delta(\sigma_0-v)
        \le V(M;v,q)+d(\sigma_0-v).
\end{align*}
If $\delta\ge1$, the last inequality follows directly from
$V(M+d;v,q)\le q(\sigma_0-v)$ and $d\ge q$.
\end{proof}

\begin{theorem}[Feasibility without aggregate budget over-reporting]
\label{thm:tsf-no-overreport}
The two-sided mechanism obtained by combining the initial-price matching construction
with the one-sided proportional hand-in-hand mechanism satisfies
incentive compatibility under \eqref{eq:tsf-budget-restriction},
including against finite collections of identities in both directions.
It also satisfies the properties in Lemma~\ref{lem:tsf-settlement}.
\end{theorem}
\begin{proof}
The settlement requirements follow as for equal-share hand-in-hand:
the shares sum to one, trading occurs only at acceptable marginal
compensation, and no initial linear segment can be left incomplete
while budget remains. The proportional updates depend only on reported
asks and remaining budgets, establishing anonymity.

Fix a true sell-$Y$ user of type $(v,q)$ and all other reports.
Before analyzing its deviations, we use the portfolio price bound
from Lemma~\ref{lem:tsf-settlement} to bound its utility.
For its collection $\playerids$ of identities, put
$n:=S_Y(\playerids)-R_Y(\playerids)$.
If $n<0$, \eqref{eq:tsf-portfolio-price-bound} gives $\Delta_X<0$,
so the trade has catastrophic utility. Every deviation of finite
utility therefore satisfies
\begin{equation}
    0\le n\le q,\qquad
    \Util^Y(\playerids;v,q)\le(\sigma_0-v)n.
    \label{eq:tsf-finite-utility-bound}
\end{equation}
This implication uses only the portfolio price bound and the model's
utility definition, not the particular one-sided allocation rule.
Since truthful utility is nonnegative, types with $q=0$ or
$v\ge\sigma_0$ cannot profitably deviate. We may therefore assume
$q>0$, $v<\sigma_0$, and a deviation of finite utility.

Let $B_{-\playerids}$ be the other users' total eligible sell-$Y$ budget and let
$M$ be $\rho_0$ times their total eligible sell-$X$ budget.
The truthful utility is $V(M;v,q)$ from
Lemma~\ref{lem:tsf-proportional-utility-bounds}, including when the truthful
order is on the minority side: if $B_{-\playerids}+q<M$, linear completion makes
the auxiliary execution on $G_M^Y$ spend every eligible sell-$Y$
budget fully at the initial price, exactly as in the two-sided
construction. Write $C\le q$ for the user's total eligible sell-$Y$
declared budget, and write $d$ for $\rho_0$ times its total eligible
sell-$X$ declared budget.

Suppose first that sell-$Y$ is the majority side after the deviation,
so $B_{-\playerids}+C\ge M+d$. The user's sell-$X$ identities are fully spent at
the initial price: together they receive $d$ units of $Y$ and sell
$\sigma_0d$ units of $X$. Its sell-$Y$ identities face $G_{M+d}^Y$.
Lemma~\ref{lem:tsf-proportional-incentives} and
\eqref{eq:tsf-proportional-prefix-bound} give
\[
    \Util^Y(\playerids;v,q)
       \le V(M+d;v,q)-(\sigma_0-v)d
       \le V(M;v,q).
\]

Suppose instead that sell-$X$ is the majority side. The user's
sell-$Y$ identities spend $C$ and receive $\sigma_0C$. Put
$e:=\rho_0 S_X(\playerids)$. Their opposite-direction receipts obey
$R_Y(\playerids)\le e$, and hence
\[
    \Util^Y(\playerids;v,q)
      =\sigma_0C-\sigma_0e-v(C-R_Y(\playerids))
      \le(\sigma_0-v)(C-e).
\]
The reverse mechanism must complete its initial linear prefix of
length $\sigma_0(B_{-\playerids}+C)$. Since the other users supply at most
$\sigma_0M$ units of $X$, this requires
$e\ge(B_{-\playerids}+C-M)_+$. It follows that
\[
    \Util^Y(\playerids;v,q)
       \le(\sigma_0-v)\bigl(C-(B_{-\playerids}+C-M)_+\bigr)
       \le(\sigma_0-v)\min\{q,(M-B_{-\playerids})_+\}
       \le V(M;v,q).
\]
The final step uses \eqref{eq:tsf-proportional-prefix-lower}.
Exchanging the tokens proves the result for true sell-$X$ users.

The mechanism's outcomes are independent of bid ordering, so the
preceding argument also applies to builders-as-users under the same
reporting restrictions. For a builder without intrinsic trading
demand, the portfolio price bound $\Delta_X+\sigma_0\Delta_Y\le0$
rules out a strict net gain in either token without a net loss in
the other, and hence rules out risk-free profit from injected orders.
\end{proof}

\subsection{Public Values: A Layered Mechanism}
\label{sec:tsf-public-values}

This subsection considers the setting in which each user's value
in its true trading direction is public, while its budget remains
private. Users may create Sybil identities and over-report budgets,
but every order in their true trading direction must report their
public value.

Specifically, a true sell-$Y$ user of type $(v,q)$ may submit any
finite collection of sell-$Y$ identities, all with ask $v$ and
arbitrary declared budgets, with no restriction on their aggregate
declared budget. It may also submit sell-$X$ identities with
arbitrary asks and declared budgets; these asks are not constrained
by its public value. The symmetric restriction applies to a true
sell-$X$ user. These are restrictions on admissible strategies;
the mechanism does not observe which identities belong to the
same user. As in the original model, the true-budget constraint
on realized net sales remains unchanged.

We first give a common-value allocation rule for a general concave
curve, including the augmented curve's linear prefix. We then assign
each value layer a baseline determined by the reported budgets at
strictly lower asks. This baseline remains fixed in the incentive
comparison even when opposite-direction identities change the
crossing amount.

\paragraph{The common-value mechanism.}
Suppose the orders in a layer $K$ share value $v$ and face a curve $g$
with $g(0)=0$ that is continuous, strictly increasing, and concave.
Write $s:=g'_+(0)\in(0,\infty)$ and define the average compensation
\[
    a(t):=\begin{cases}g(t)/t,&t>0,\\s,&t=0.\end{cases}
\]
For the budget reports of this layer, put
\[
    B:=\sum_{i\in K}\reportq{i},\qquad
    D_i:=\sum_{j\in K\setminus\{i\}}\reportq{j},\qquad
    \phi_i(r):=r\bigl(a(D_i+r)-v\bigr).
\]
Choose the largest maximizer
\begin{equation*}
    s_i:=\max\arg\max_{0\le r\le\reportq{i}}\phi_i(r),
    \qquad \tau_i:=\phi_i(s_i),\qquad
    Q:=\sum_{i\in K}s_i,\qquad U:=\sum_{i\in K}\tau_i.
\end{equation*}
Here $s_i$ is a critical input quantity and $\tau_i$ is the order's
target utility; neither denotes a received token amount. Let $P$ be
the largest solution in $[Q,B]$ of
\begin{equation}
    g(P)-vP=U.
    \label{eq:tsf-public-clearing}
\end{equation}
If $B>Q$, assign
\begin{equation*}
    y_i=s_i+\frac{P-Q}{B-Q}(\reportq{i}-s_i),
    \qquad x_i=\tau_i+vy_i.
\end{equation*}
If $B=Q$, instead set $y_i=\reportq{i}$ and use the same formula for
$x_i$. In particular, an empty or zero-budget layer receives zero.

\begin{lemma}[Properties of the common-value mechanism]
\label{lem:tsf-public-kernel}
The common-value mechanism is well defined and satisfies the
one-sided settlement requirements, including linear completion.
Holding $g$ and the other reports fixed, a true user with value $v$
and budget $q$ cannot gain from budget misreports or same-value Sybil
identities whenever its aggregate realized input is at most $q$.
\end{lemma}
\begin{proof}
Concavity makes $a$ nonincreasing and continuous, including at zero.
The critical maxima exist on compact intervals and are nonnegative.
Write $W(t):=g(t)-vt$. Evaluating each critical objective at the full
reported budget gives
\[
    W(B)=\sum_i\reportq{i}\bigl(a(B)-v\bigr)
         =\sum_i\phi_i(\reportq{i})\le U.
\]
Also
$D_i+s_i\ge Q$, so monotonicity of $a$ gives
\[
    U\le\sum_i s_i(a(Q)-v)=W(Q).
\]
When $Q=0$, both sides of this last inequality are zero.
Continuity now gives a solution of \eqref{eq:tsf-public-clearing}
in $[Q,B]$, and the solution set is compact. Thus its largest element
exists. The allocations satisfy $s_i\le y_i\le\reportq{i}$,
$\sum_i y_i=P$, and $\sum_i x_i=U+vP=g(P)$.

Every order has utility $\tau_i\ge0$. If $v\le s$, then
$\tau_i\le(s-v)s_i\le(s-v)y_i$, so $x_i\le sy_i$.
If $v>s$, every critical quantity is zero and the only clearing
solution is zero, giving the same bound. If $P<B$ and
$g'_+(P)>v$, the function $W$ is strictly above $U$ just to the right
of $P$. Since $W(B)\le U$, there must be a later clearing solution,
a contradiction. An individually unfilled order implies $P<B$,
so this proves weak local efficiency. The allocation formulas depend
only on the common value and reported budgets, establishing anonymity.

If $g(t)=st$ for $0\le t\le L$ and $v\le s$, a terminal point below
$\ell:=\min\{B,L\}$ is impossible. For $v<s$ it contradicts the
preceding efficiency argument. For $v=s$, concavity gives $W(t)\le0$
for all $t$, so $0\le U\le W(Q)$ implies $U=W(Q)=0$. If $Q\ge\ell$,
then $P\ge Q\ge\ell$. If $Q<\ell$, then $\ell\in[Q,B]$ and
$W(\ell)=0=U$, so $\ell$ is an admissible clearing solution and the
largest-solution convention gives $P\ge\ell$. This argument does not
require $L$ to be the maximal length of an initial linear segment.
If $B\le L$, the largest-maximizer convention also gives
$s_i=\reportq{i}$, hence full-budget initial-price outcomes.
When $v=0$, the function $r\mapsto r g(D_i+r)/(D_i+r)$ is
strictly increasing, with value zero at $r=0$; hence every critical
maximizer is its full reported cap, and every order is fully spent.

For incentives, let $D$ be the total declared budget of the other
orders in this layer. The truthful user's utility is
\[
    \max_{0\le r\le q}r\bigl(a(D+r)-v\bigr).
\]
For a deviating collection $I\subseteq K$ of same-value identities,
put $Q_I:=\sum_{i\in I}s_i$. Their realized input is at least $Q_I$,
so a deviation with realized input at most $q$ has $Q_I\le q$.
Moreover, $D_i+s_i\ge D+Q_I$ for every $i\in I$. Therefore
\[
    \sum_{i\in I}(x_i-vy_i)=\sum_{i\in I}\tau_i
       \le Q_I\bigl(a(D+Q_I)-v\bigr)
       \le\max_{0\le r\le q}r\bigl(a(D+r)-v\bigr).
\]
This proves the incentive claim without any restriction on the sum
of the deviating declared budgets.
\end{proof}

\paragraph{Layer baselines determined by reports.}
Group equal public asks into layers, with lower asks having
priority. 
On a one-sided effective curve $G$, define
\begin{equation*}
    K_v:=\{i:\reportv{i}=v\},\qquad
    B_v:=\sum_{i\in K_v}\reportq{i},\qquad
    C_v:=\sum_{i:\,\reportv{i}<v}\reportq{i}.
\end{equation*}
Only initially eligible orders are included. The baseline $C_v$ is
the sum of \emph{reported budgets} at strictly lower asks, not the
realized input of earlier layers. True budgets remain private, and
declared budgets may be changed or over-reported. In an incentive
comparison for a user with public ask $v$, this baseline stays unchanged
because none of its same-value identities enters the sum defining $C_v$.
The baseline can change when other users change their reports.
Declare the value-$v$ layer active if
\begin{equation}
    v\le G'_+(C_v).
    \label{eq:tsf-public-activity}
\end{equation}
An inactive layer receives zero. For an active layer, run the
common-value mechanism on $K_v$ using the residual curve
\begin{equation*}
    g_v(t):=G(C_v+t)-G(C_v).
\end{equation*}
Let $P_v$ denote its total realized input, and let $P:=\sum_v P_v$.
These definitions specify each layer independently from the report
profile. Equality in \eqref{eq:tsf-public-activity} is retained so that
orders at the initial price can fill a linear crossing prefix.

The next lemma explains why reported-budget baselines are consistent
with actual settlement: all active layers except possibly the last
spend their full reported budgets.

\begin{lemma}[Structure and settlement of the value layers]
\label{lem:tsf-public-structure}
Active asks form an initial segment of the increasing order of
public asks. Every active layer strictly below the largest active
ask spends its full reported budget. The resulting one-sided rule
satisfies \eqref{eq:tsf-kernel-requirements},
\eqref{eq:tsf-linear-completion}, and anonymity.
\end{lemma}
\begin{proof}
If $v$ is active and $u<v$, then $C_u\le C_v$ and concavity gives
\[
    G'_+(C_u)\ge G'_+(C_v)\ge v>u.
\]
Thus $u$ is active. Next suppose $u<v$ are both active and the
$u$-layer leaves an order unfilled. The common-value mechanism's
local efficiency gives
\[
    u\ge G'_+(C_u+P_u).
\]
But $P_u\le B_u$ and $C_u+B_u\le C_v$, so
\[
    G'_+(C_u+P_u)\ge G'_+(C_v)\ge v,
\]
a contradiction. Hence every active layer before the largest active
ask is fully spent.

If there is an active layer, let $v_*$ be the largest active ask.
All smaller-ask layers are fully spent and all larger-ask layers
receive zero. Consequently
\[
    P=C_{v_*}+P_{v_*}.
\]
The output increments of the fully spent layers telescope to
$G(C_{v_*})$; the last layer adds
$G(C_{v_*}+P_{v_*})-G(C_{v_*})$. Thus total output is $G(P)$.
If there is no active layer, both totals are zero.
The reported-budget constraints and individual rationality carry over
immediately from each layer. The output rate of each layer is at most
$G'_+(C_v)\le G'_+(0)$. Layers and their baselines depend only on
reported asks and budgets, so anonymity follows from that of the
common-value mechanism.

For global weak local efficiency, suppose some order is unfilled,
and let $u$ be the smallest ask with this property. All preceding
layers are fully spent. If the $u$-layer is active, it is the last
active layer, so $P=C_u+P_u$ and its local efficiency gives
$u\ge G'_+(P)$. If it is inactive, all later layers are inactive;
since every preceding layer is fully spent, $P=C_u$. Inactivity
then gives $u>G'_+(P)$. In either case every
unfilled order has ask at least $G'_+(P)$. This argument does not
assume that an inactive layer's reported baseline equals the
terminal input.

Finally, suppose $G(t)=\gamma_0t$ on $[0,L]$, and all input asks
are at most $\gamma_0$. If $P<\min\{L,\sum_v B_v\}$, let $u$ again
be the first unfilled layer. Its baseline satisfies $C_u\le P<L$,
so it is active. Its residual curve is linear up to $L-C_u$.
The common-value mechanism therefore gives
$P_u\ge\min\{B_u,L-C_u\}$, contradicting both $P_u<B_u$ and
$C_u+P_u=P<L$. This proves linear completion.
\end{proof}

\paragraph{A user's truthful utility.}
Fix the other eligible sell-$Y$ reports and a public ask $v$.
Let $C$ be their total budget at strictly lower asks and $D$ their
total budget at ask $v$. Neither quantity changes when the user
changes its own budget, splits into same-value identities, or
submits opposite-direction identities. On the augmented curve, put
\[
    g_{M,C}(t):=G_M^Y(C+t)-G_M^Y(C).
\]
The truthful utility of $Y(v,q)$ is
\begin{equation}
    T(M;v,q)
       :=\max_{0\le r\le q}
           r\left(\frac{g_{M,C}(D+r)}{D+r}-v\right).
    \label{eq:tsf-public-truthful-utility}
\end{equation}
At $D+r=0$, the average compensation is interpreted as
$g'_{M,C,+}(0)$. If the layer is active, this is its critical utility.
If the layer is inactive, its initial marginal compensation is
below $v$, so the displayed maximum is zero and still equals its
actual utility.

For same-direction deviations the entire residual curve is fixed.
The common-value incentive argument therefore applies, including to
arbitrary aggregate budget over-reports whose realized sale respects
the true budget. Opposite-direction identities also change $M$.
We next bound the benefit from this change.

\begin{lemma}[The value of additional crossing under public values]
\label{lem:tsf-public-crossing}
Fix any finite profile of other eligible sell-$Y$ reports, and write
$A$ for their total budget. For all $M,d,q\ge0$ and
$0\le v\le\sigma_0$, the truthful utility in
\eqref{eq:tsf-public-truthful-utility} satisfies
\begin{align}
    T(M+d;v,q+d)
       &\le T(M;v,q)+(\sigma_0-v)d,
       \label{eq:tsf-public-crossing-bound}\\
    T(M;v,q)
       &\ge(\sigma_0-v)\min\{q,(M-A)_+\}.
       \label{eq:tsf-public-prefix-lower}
\end{align}
\end{lemma}
\begin{proof}
Keep the reported-budget baseline $C$ and same-value outside budget
$D$ fixed. Extending the initial-price prefix gives, for $t\ge0$,
\[
    G_{M+d}^Y(t+d)=\sigma_0d+G_M^Y(t),
    \qquad G_{M+d}^Y(t)\ge G_M^Y(t).
\]
The second inequality follows from concavity: replacing a portion
of the compensation curve by initial-price trades cannot reduce
output. Set
\[
    \kappa:=G_{M+d}^Y(C)-G_M^Y(C)\ge0.
\]
The two residual curves then satisfy the exact identity
\begin{equation}
    g_{M+d,C}(t+d)=g_{M,C}(t)+\sigma_0d-\kappa.
    \label{eq:tsf-public-residual-shift}
\end{equation}

Let $\Phi_M(r)$ denote the objective in
\eqref{eq:tsf-public-truthful-utility}. For any candidate
$r'\in[0,q+d]$, if $r'<d$, the initial-price bound gives
\[
    \Phi_{M+d}(r')\le(\sigma_0-v)r'
       \le(\sigma_0-v)d
       \le T(M;v,q)+(\sigma_0-v)d.
\]
Otherwise write $r'=r+d$, where $0\le r\le q$, and set $t:=D+r$.
For $t>0$, using \eqref{eq:tsf-public-residual-shift} gives
\begin{align*}
    &\Phi_M(r)+(\sigma_0-v)d-\Phi_{M+d}(r+d)\\
    &\qquad=
       \frac{dD\bigl(\sigma_0t-g_{M,C}(t)\bigr)}{t(t+d)}
       +\frac{(r+d)\kappa}{t+d}\ge0.
\end{align*}
Here $g_{M,C}(t)\le\sigma_0t$. If $t=0$, then $D=r=0$ and the
same conclusion follows directly from the initial-price bound at
$r'=d$, including $d=0$. Maximizing over $r'$ proves
\eqref{eq:tsf-public-crossing-bound}. In particular, the bound also
covers a layer that is inactive before the prefix is extended.

For the lower bound, $C+D\le A$. If $M>A$, the critical quantity
\[
    r_0:=\min\{q,M-A\}
\]
satisfies $C+D+r_0\le M$. The relevant residual segment lies entirely
inside the initial-price prefix, so
$\Phi_M(r_0)=(\sigma_0-v)r_0$. If $M\le A$, nonnegative truthful
utility gives the required zero lower bound.
\end{proof}

\begin{theorem}[Feasibility with public values]
\label{thm:tsf-public-values}
The two-sided mechanism obtained by combining the initial-price matching construction
with the one-sided public-value layered mechanism satisfies incentive
compatibility under the public-value strategy restriction stated above. This includes
arbitrary finite collections of identities in both directions and
arbitrary aggregate budget over-reporting, with no restriction on
the asks of opposite-direction identities. The mechanism also
satisfies the properties in Lemma~\ref{lem:tsf-settlement}.
\end{theorem}
\begin{proof}
Lemma~\ref{lem:tsf-public-structure} verifies the one-sided
requirements, so Lemma~\ref{lem:tsf-settlement} gives the settlement
properties and the portfolio price bound.

Fix a true sell-$Y$ user of type $(v,q)$ and all other reports.
The portfolio-bound argument giving \eqref{eq:tsf-finite-utility-bound}
applies here as well: every deviation of finite utility has net $Y$
sale $n\in[0,q]$ and utility at most $(\sigma_0-v)n$.
Together with truthful individual rationality, this rules out a
profitable deviation when $q=0$ or $v\ge\sigma_0$. It therefore
suffices to consider $q>0$, $v<\sigma_0$, and a deviation of finite
utility.

Let $A$ be the other users' total eligible sell-$Y$ budget and let
$M$ be $\rho_0$ times their total eligible sell-$X$ budget.
Truthful utility is $T(M;v,q)$ even when the truthful order is on
the minority side. Indeed, if $A+q<M$, linear completion makes
the one-sided execution on $G_M^Y$ give every eligible sell-$Y$
order its full-budget initial-price outcome, exactly as in the
two-sided construction.

Suppose first that sell-$Y$ is the majority side after the
deviation. Write
\[
    d:=\rho_0\sum_{j\in\playerids\cap\Xside_0}\reportr{j}.
\]
The user's eligible sell-$X$ identities receive $d$ units of $Y$
and spend $\sigma_0d$ units of $X$. Its sell-$Y$ identities face
$G_{M+d}^Y$, and their aggregate realized input is at most $q+d$.
They all have ask $v$, so the lower-ask reported baseline $C$
and same-value outside budget $D$ are unchanged. If the layer is
active, the common-value incentive argument with cap $q+d$ gives
\[
    \sum_{i\in\playerids\cap\Yside}(x_i-vy_i)
       \le T(M+d;v,q+d).
\]
This uses the bound on aggregate \emph{realized} input, not a bound
on the sum of declared budgets. If the layer is inactive, its
identities receive zero, and $T(M+d;v,q+d)=0$ because activity
depends on $C$ and $v$, not on the user's same-value budget reports.
The opposite-direction trades
contribute $-(\sigma_0-v)d$ to utility, and
\eqref{eq:tsf-public-crossing-bound} yields
\[
    \Util^Y(\playerids;v,q)
       \le T(M+d;v,q+d)-(\sigma_0-v)d
       \le T(M;v,q).
\]

Suppose instead that sell-$X$ is the majority side. Let $b$ be the
user's aggregate eligible sell-$Y$ declared budget, which is now
fully spent at the initial price. Put
\[
    e:=\rho_0 S_X(\playerids),\qquad
    r:=R_Y(\playerids)\le e.
\]
The initial-price bound gives the inequality $r\le e$. Its utility
therefore satisfies
\[
    \Util^Y(\playerids;v,q)
       =\sigma_0 b-\sigma_0 e-v(b-r)
       \le(\sigma_0-v)(b-e).
\]
The reverse mechanism must complete a crossing prefix of length
$\sigma_0(A+b)$. Since the other users can supply at most
$\sigma_0M$ units of $X$, this requires
$e\ge(A+b-M)_+$. Hence
\[
    b-e\le b-(A+b-M)_+
         =\min\{b,M-A\}\le M-A.
\]
Consequently
\[
    \Util^Y(\playerids;v,q)
       \le(\sigma_0-v)(M-A).
\]
Independently, \eqref{eq:tsf-finite-utility-bound} gives utility
at most $(\sigma_0-v)q$. These two bounds imply
\[
    \Util^Y(\playerids;v,q)
       \le(\sigma_0-v)\min\{q,(M-A)_+\}
       \le T(M;v,q),
\]
where the final inequality is
\eqref{eq:tsf-public-prefix-lower}. This case does not assume
$b\le q$: the public-value setting permits arbitrary declared
budgets, and only net realized expenditure is constrained.

Exchanging $X$ and $Y$, and using $H$ and $\rho_0$ in place of
$F$ and $\sigma_0$, proves the result for true sell-$X$ users.
\end{proof}

\section*{AI Acknowledgements}
This work began in December, 2025. 
The authors developed the shared two-sided matching construction and the first two mechanisms: 
one for the setting without Sybil identities and one for the setting without aggregate budget over-reporting.
The authors also initially proved an impossibility result under a stronger efficiency requirement than weak local efficiency (wLE). 
In July 2026, ChatGPT 5.6 developed the mechanism for the public-value setting, but failed to strengthen the impossibility result. 
ChatGPT-6-Astra subsequently produced a proof establishing impossibility under wLE. 
The authors carefully reviewed and substantially revised all AI-generated text.

\bibliographystyle{alpha}
\bibliography{refs}

\end{document}